\documentclass{article}
\usepackage{amsmath,amsthm,amssymb}
\usepackage{enumerate,graphicx}

\usepackage{doi,url}

\def\titlerunning#1{\gdef\titrun{#1}}
\makeatletter
\def\author#1{\gdef\autrun{\def\and{\unskip, }#1}\gdef\@author{#1}}
\def\address#1{{\def\and{\\\hspace*{18pt}}\renewcommand{\thefootnote}{}%
\footnote {#1}}%
\markboth{\autrun}{\titrun}}
\makeatother
\def\email#1{E-mail: #1}
\def\subjclass#1{\par\medskip
\noindent\textbf{Mathematics Subject Classification (2010).} #1}
\def\keywords#1{\par\medskip
\noindent\textbf{Keywords.} #1}

\newenvironment{acknowledgement}{\textbf{Acknowledgement.}}

\newtheorem{theorem}{Theorem}[section]
\newtheorem{proposition}[theorem]{Proposition}
\newtheorem{lemma}[theorem]{Lemma}

\theoremstyle{definition}
\newtheorem{definition}[theorem]{Definition}

\newtheorem{remark}[theorem]{Remark}

\numberwithin{equation}{section}

\DeclareMathOperator{\supp}{supp}
\DeclareMathOperator{\tr}{tr}

\DeclareMathOperator{\dist}{dist}

\DeclareMathOperator{\diam}{diam}

\makeatletter
\DeclareRobustCommand\widecheck[1]{{\mathpalette\@widecheck{#1}}}
\def\@widecheck#1#2{%
    \setbox\z@\hbox{\m@th$#1#2$}%
    \setbox\tw@\hbox{\m@th$#1%
       \widehat{%
          \vrule\@width\z@\@height\ht\z@
          \vrule\@height\z@\@width\wd\z@}$}%
    \dp\tw@-\ht\z@
    \@tempdima\ht\z@ \advance\@tempdima2\ht\tw@ \divide\@tempdima\thr@@
    \setbox\tw@\hbox{%
       \raise\@tempdima\hbox{\scalebox{1}[-1]{\lower\@tempdima\box
\tw@}}}%
    {\ooalign{\box\tw@ \cr \box\z@}}}
\makeatother

\newcommand{\pr}{\prime}

\newcommand\R{\mathbb R}
\newcommand\N{\mathbb N}
\newcommand\C{\mathbb C}
\newcommand\Z{\mathbb Z}

\renewcommand\P{\mathbb P}

\newcommand\e{\mathrm{e}}

\newcommand\eps{\varepsilon}
\newcommand\vphi{\varphi}

\newcommand\La{\Lambda}

\newcommand{\abs}[1]{\left\lvert #1 \right\rvert}
\newcommand{\norm}[1]{\left\lVert #1 \right\rVert}

\newcommand{\set}[1]{\left\{ #1 \right\}}

\newcommand\beq{\begin{equation}}
\newcommand\eeq{\end{equation}}

\newcommand{\qtx}[1]{\quad\text{#1}\quad}
\newcommand{\sqtx}[1]{\;\text{#1}\;}

\begin{document}
\titlerunning{Eigensystem Bootstrap Multiscale Analysis}
\title{Eigensystem Bootstrap Multiscale Analysis for the Anderson Model}

\author{Abel Klein\thanks{A.K. was partially  supported  by the NSF under grant DMS-1301641.} \and   C.S. Sidney Tsang
\thanks{C.S.S.T.  was supported  by the NSF under grant DMS-1301641.}}

\date{}
\maketitle

\address{{University of California, Irvine;
Department of Mathematics;
Irvine, CA 92697-3875,  USA.}
\email{aklein@uci.edu, tsangcs@uci.edu.}}

\begin{abstract}
We use a bootstrap argument to enhance the eigensystem multiscale analysis, introduced by Elgart and Klein for proving localization for the Anderson model at high disorder.  The eigensystem multiscale analysis studies finite volume
eigensystems, not finite volume  Green's functions. It yields pure point spectrum with exponentially decaying eigenfunctions and dynamical localization. The starting hypothesis for the eigensystem bootstrap multiscale analysis only requires the verification of
polynomial decay of the finite volume eigenfunctions, at some sufficiently large scale, with some minimal probability independent of the scale. It yields exponential
localization of finite volume eigenfunctions in boxes of side $L$, with the eigenvalues and eigenfunctions labeled by the sites of the box, with probability higher
than $1-\mathrm{e}^{-L^\xi}$, for any desired $0<\xi<1$.
\end{abstract}

\tableofcontents

\addcontentsline{toc}{section}{Introduction}
\section*{Introduction}

The  eigensystem multiscale analysis  is a new approach for proving  localization  for the Anderson model  introduced by Elgart and Klein \cite{EK}.
The usual proofs of localization for random Schr\"odinger operators are based on the study of finite volume Green's functions
\cite{FS,FMSS,Dr,DK,Sp,CH,FK,GKboot,Kle,BK,GKber,AM,A,ASFH,AENSS}. In contrast to the usual strategy, the eigensystem multiscale analysis is based on finite volume
 eigensystems, not finite volume Green's functions. It treats all energies of the finite volume operator at the same time, establishing level spacing and  localization of  eigenfunctions in a fixed box with high probability.
  A new feature is the  labeling of the eigenvalues and eigenfunctions  by the sites of the box.

In this paper we use a bootstrap argument as in Germinet and Klein \cite{GKboot} to enhance the eigensystem multiscale analysis. It yields exponential localization of finite volume eigenfunctions
in boxes of side $L$, with the eigenvalues and eigenfunctions labeled by the sites of the box, with probability higher than $1-\mathrm{e}^{-L^\xi}$, for any $0<\xi<1$. The starting
hypothesis for the eigensystem bootstrap multiscale analysis only requires the verification of polynomial decay of the finite volume eigenfunctions, at some sufficiently large scale,
with some minimal probability independent of the scale. The advantage of the bootstrap multiscale analysis is that from the same starting hypothesis we get conclusions that are valid for any $0<\xi<1$.

We consider the Anderson model  $H_{\eps,\omega}=-\varepsilon\Delta+V_\omega$ on  
 $\ell^2(\Z^d)$ (see Definition~\ref{defanderson}; $\eps >0$ is the inverse of the disorder parameter).  Multiscale analyses study finite volume operators $H_{\eps,\omega,\La}$, the restrictions  of $H_{\eps,\omega}$ to  finite  boxes $\La$.  The objects of interest for the  eigensystem multiscale analysis are finite volume
 eigensystems. An eigensystem $\set{(\vphi_j,\lambda_j)}_{j\in J}$ for $H_{\eps,\omega,\La}$ consists of eigenpairs 
$(\vphi_j,\lambda_j)$, where  $\lambda_j$ is an eigenvalue for  $H_{\eps,\omega,\La}$  and $\vphi_j$ is a corresponding normalized eigenfunction,  such that $\set{\vphi_j}_{j\in J}$ is an orthonormal basis for the finite dimensional Hilbert space $\ell^2(\La)$.  
Elgart and Klein \cite{EK} called a box $\La$   localizing for $H_{\eps,\omega}$  if the eigenvalues  of $H_{\eps,\omega,\La}$ satisfy a level spacing condition, and there exists an eigensystem for $H_{\eps,\omega,\La}$ indexed by the sites in the box,  
 $\set{(\vphi_x,\lambda_x)}_{x\in\La}$, with the eigenfunctions
 $\set{\vphi_x}_{x\in\La}$  exhibiting exponential localization around the label, i.e.,
$\abs{\vphi_x(y)} \le \e^{-m\norm{x-y}}$ for $y\in \La$ distant from $x$. They  showed \cite[Theorem~1.6]{EK} that, fixing $\xi \in (0,1)$,   at high disorder ($\eps \ll1$) boxes of (sufficiently large) side $L$ are localizing with probability $\ge 1 - \e^{-L^\xi}$, yielding all the usual forms of localization \cite[Theorem~1.7 and Corollary~1.8]{EK}.  More precisely, it is shown in \cite{EK} that for $\xi \in (0,1)$ there exists $\eps_\xi >0$, decreasing as $\xi$ increases, 
and for $\eps>0$ a scale $L_\eps$, increasing as $\eps$ decreases, such that for  $0<\eps \le \eps_\xi$ and  $L\ge L_{\eps_\xi}$ boxes of side $L$ are localizing for $H_{\eps,\omega}$ with probability $\ge 1 - \e^{-L^\xi}$.

 We use the ideas of  Germinet and Klein \cite{GKboot}   to perform a bootstrap multiscale analysis for finite volume  eigensystems (Theorem~\ref{mainthm}).  To start the multiscale analysis, we only have to verify a statement
 of polynomial localization  of the eigenfunctions with some minimal probability independent of the scale.  We conclude  that at high disorder boxes of side $L$ are localizing with probability $\ge 1 - \e^{-L^\xi}$ for all  $\xi \in (0,1)$.   It follows  (Theorem~\ref{maincor}) that there exists $\eps_0>0$, and for each  $\xi \in (0,1)$ there exists a scale $L_{\eps_0,\xi}$, such that for all  $0<\eps \le \eps_0$ and  $L\ge L_{\eps_0,\xi}$ boxes of side $L$ are localizing for $H_{\eps,\omega}$ with probability $\ge 1 - \e^{-L^\xi}$.  How large $L$ needs to be  depends on $\xi$, 
  but the required amount of disorder is independent of $\xi$.  In addition, if we have the conclusions of \cite[Theorem~1.6]{EK} for a fixed $\xi \in (0,1)$,  it follows from Theorem~\ref{mainthm} that for all $\xi^\pr \in (0,1)$ there exists a scale $L_{\xi^\pr}$, such that for all $0<\eps \le \eps_\xi$  and  $L\ge L_{\xi^\pr}$ boxes of side $L$ are localizing for $H_{\eps,\omega}$ with probability $\ge 1 - \e^{-L^{\xi^\pr}}$. (Note that $\eps_\xi$ depends on the fixed $\xi$ but does not depend on $\xi^\pr$.)

  Recently,  Elgart and Klein \cite{EK2} extended  the eigensystem multiscale analysis   to  establish  localization   for the Anderson model in an energy interval. This extension yields  localization at fixed disorder  on an interval at the edge of the spectrum (or in  the vicinity of a spectral gap), and  at  a fixed interval at the bottom of the spectrum  for sufficiently  high disorder.  We expect that our bootstrap eigensystem multiscale analysis can also be extended to energy intervals.

Our main definitions and resuts are stated in Section~\ref{secmaindr}.  Theorem~\ref{mainthm} is the bootstrap eigensystem multiscale analysis. Theorem~\ref{maincor} gives the high disorder result for the Anderson model, and yields Theorem~\ref{thmloc}, which encapsulates localization for the Anderson model at high disorder.  Theorem~\ref{mainthm} is proven  in Section~\ref{sectbmsa}, and Theorem~\ref{maincor} is proven in Section~\ref{sectinit}.  In Section~\ref{secprel} we provide notation,  definitions and lemmas  for the proof of the  bootstrap eigensystem multiscale analysis.  In Section~\ref {secprob} we state the probability estimates for level spacing used in the proof of the  bootstrap eigensystem multiscale analysis.

\section{Main definitions and  results}\label{secmaindr}
We consider the Anderson model in the following form.

\begin{definition}\label{defanderson}
The Anderson model is the random Schr\"odinger operator
\beq
H_{\varepsilon,\omega}:=-\varepsilon\Delta+V_\omega\qtx{on}\ell^2(\Z^d),
\eeq
where $\varepsilon>0$;
$\Delta$ is the (centered) discrete Laplacian:
 \beq
(\Delta\varphi)(x):=\sum_{y\in\Z^d,|y-x|=1}\varphi(y)\qtx{for}\varphi\in\ell^2(\Z^d);
\eeq
$V_\omega(x)=\omega_x$ for $x\in\Z^d$, where
$\omega=\{\omega_x\}_{x\in\Z^d}$ is a family of independent identically distributed random
variables, with a non-degenerate probability distribution $\mu$ with bounded support and H\"older continuous of order $\alpha\in(\frac{1}{2},1]$:
\beq
S_\mu(t)\leq Kt^\alpha\qtx{for all}t\in[0,1],
\eeq
 with $S_\mu(t):=\sup_{a\in\R}\mu\{[a, a+t]\}$  the concentration function of the measure $\mu$ and $K$  a constant.
\end{definition}

Given $\Theta\subset\Z^d$, we let $T_\Theta=\chi_\Theta T\chi_\Theta$ be the restriction of the bounded operator $T$ on  $\ell^2(\Z^d)$ to $\ell^2(\Theta)$.
If  $\Phi\subset\Theta\subset\Z^d$,  we identify $\ell^2(\Phi)$ with a subset of $\ell^2(\Theta)$ by extending functions on $\Phi$ to functions on $\Theta$ that are identically $0$ on $\Theta\setminus\Phi$.
We write $\varphi_{\Phi}=\chi_{\Phi}\varphi$ if $\varphi$ is a function on $\Theta$. We let $\|\varphi\|=\|\varphi\|_2$ and  $\|\varphi\|_\infty=\max_{y\in\Theta}|\varphi(y)|$
for $\varphi\in\ell^2(\Theta)$.

For  $x=(x_1,x_2,\ldots,x_d)\in\R^d$ we set $\|x\|=|x|_\infty= \max_{j=1,2,\ldots,d}|x_j|$,
$|x|=|x|_2=\left(\sum_{j=1}^dx_j^2\right)^{\frac{1}{2}}$, and $|x|_1=\sum_{j=1}^d|x_j|$.  Given $\Xi\subset\R^d$, we let $\diam\Xi=\sup_{x,y\in\Xi}\|y-x\|$ denote its diameter, and  set $\dist(x,\Xi)=\inf_{y\in\Xi}\|y-x\|$ for $x\in\R^d$.

We use boxes in $\Z^d$ centered at points in $\R^d$. The box in $\Z^d$ of side $L>0$ centered at $x\in\R^d$ is given by
\beq
\Lambda_L(x)=\Lambda^\R_L(x)\cap\Z^d, \qtx{where} \Lambda^\R_L(x)=\left\{y\in\R^d;\|y-x\|\leq\tfrac{L}{2}\right\}.
\eeq
We write  $\Lambda_L$ to denote a box $\Lambda_L(x)$ for some $x\in\R^d$. We have
$(L-2)^d<|\Lambda_L|\leq(L+1)^d$ for $L\geq2$, where for  a set $\Theta \subset \Z^d$ we let $\abs{\Theta}$ denote its cardinality.

The following definitions are for a fixed discrete Schr\"odinger operator $H_\varepsilon$. We omit $\varepsilon$ from the notation (i.e., we write $H$ for $H_\varepsilon$, $H_\Theta$
for $H_{\varepsilon,\Theta}$) when it does not lead to confusion. We always consider scales $L\geq200$, and, for $\tau\in(0,1)$, set
\beq
L'=\left\lfloor\tfrac{L}{20}\right\rfloor\qtx{and}L_\tau=\lfloor L^\tau\rfloor.
\eeq

For fixed $q>0$, $\beta,\tau\in(0,1)$, we have the following definitions:
\begin{definition}\label{defallloc}
Let $\Lambda_L$ be a box, $x\in\Lambda_L$, and $\varphi\in\ell^2(\Lambda_L)$ with $\|\varphi\|=1$. Then:
\begin{enumerate}
\item[(i)]Given $\widetilde{\theta}>0$, $\varphi$ is said to be $(x,\widetilde{\theta})$-polynomially localized if
\beq\label{defploc}
|\varphi(y)|\leq L^{-\widetilde{\theta}}\quad\text{for all}\quad y\in\Lambda_L\qtx{with}\|y-x\|\geq L'.
\eeq
\item[(ii)]Given $\widetilde{s}\in(0,1)$, $\varphi$ is said to be $(x,\widetilde{s})$-subexponentially localized if
\beq\label{defsloc}
|\varphi(y)|\leq\e^{-L^{\widetilde{s}}}\quad\text{for all}\quad y\in\Lambda_L\qtx{with}\|y-x\|\geq L'.
\eeq
\item[(iii)]Given $m>0$, $\varphi$ is said to be $(x,m)$-localized if
\beq\label{defloc}
|\varphi(y)|\leq\e^{-m\|y-x\|}\quad\text{for all}\quad y\in\Lambda_L\qtx{with}\|y-x\|\geq L_\tau.
\eeq
\end{enumerate}
\end{definition}

\begin{definition}
Let $R>0$, and $\Theta\subset\Z^d$ be a finite set such that all eigenvalues of $H_\Theta$ are simple (i.e., $|\sigma(H_\Theta)|=|\Theta|$). Then:
\begin{enumerate}
\item[(i)]$\Theta$ is called $R$-polynomially level spacing for $H_\Theta$ if $|\lambda-\lambda'|\geq R^{-q}$ for all $\lambda,\lambda'\in\sigma(H_\Theta),\lambda\neq\lambda'$.
\item[(ii)]$\Theta$ is called $R$-level spacing for $H_\Theta$ if $|\lambda-\lambda'|\geq\e^{-R^\beta}$ for all $\lambda,\lambda'\in\sigma(H_\Theta),\lambda\neq\lambda'$.
\end{enumerate}
When $\Theta=\Lambda_L$, a box,  and $R=L$, we will just say that $\Lambda_L$ is polynomially level spacing for $H_{\Lambda_L}$, or  $\Lambda_L$ is level spacing for $H_{\Lambda_L}$.\end{definition}

Note that $R$-polynomially level spacing implies $R$-level spacing for sufficiently large $R$.

Given $\Theta\subset\Z^d$, $(\varphi,\lambda)$ is called an eigenpair for $H_\Theta$ if $\varphi\in\ell^2(\Theta)$, $\lambda\in\R$ with $\|\varphi\|=1$, and $H_\Theta\varphi=\lambda\varphi$
(i.e., $\lambda$ is an eigenvalue for $H_\Theta$ with a corresponding normalized eigenfunction $\varphi$). A collection $\{(\varphi_j,\lambda_j)\}_{j\in J}$ of eigenpairs for $H_\Theta$
is called an eigensystem for $H_\Theta$ if $\{\varphi_j\}_{j\in J}$ is an orthonormal basis for $\ell^2(\Theta)$. We may rewrite the eigensystem as $\{(\psi_\lambda,\lambda)\}_{\lambda\in\sigma(H_\Theta)}$
if all eigenvalues of $H_\Theta$ are simple.

\begin{definition}\label{deflocbox}
Let $\Lambda_L$ be a box. Then:

\begin{enumerate}
\item
Given $\widetilde{\theta}>0$, $\Lambda_L$ will be called $\widetilde{\theta}$-polynomially localizing (PL) for $H$ if the following holds:
\begin{enumerate}
\item $\Lambda_L$ is polynomially level spacing for $H_{\Lambda_L}$.
\item There exists a $\widetilde{\theta}$-polynomially localized eigensystem for $H_{\Lambda_L}$, that is, an eigensystem $\{(\varphi_x,\lambda_x)\}_{x\in\Lambda_L}$ for $H_{\Lambda_L}$ such that $\varphi_x$
is $(x,\widetilde{\theta})$-polynomially localized for all $x\in\Lambda_L$.
\end{enumerate}

\item
Given $m^\ast>0$, $\Lambda_L$ will be called $m^\ast$-mix localizing (ML) for $H$ if the following holds:
\begin{enumerate}
\item$\Lambda_L$ is polynomially level spacing for $H_{\Lambda_L}$.
\item There exists an $m^\ast$-localized eigensystem for $H_{\Lambda_L}$, that is, an eigensystem $\{(\varphi_x,\lambda_x)\}_{x\in\Lambda_L}$ for $H_{\Lambda_L}$ such that $\varphi_x$
is $(x,m^\ast)$-localized for all $x\in\Lambda_L$.
\end{enumerate}

\item
Given $\widetilde{s}\in(0,1)$, $\Lambda_L$ will be called $\widetilde{s}$-subexponentially localizing (SEL) for $H$ if the following holds:
\begin{enumerate}
\item $\Lambda_L$ is level spacing for $H_{\Lambda_L}$.
\item There exists an $\widetilde{s}$-subexponentially localized eigensystem for $H_{\Lambda_L}$, that is, an eigensystem $\{(\varphi_x,\lambda_x)\}_{x\in\Lambda_L}$ for $H_{\Lambda_L}$ such that $\varphi_x$
is $(x,\widetilde{s})$-subexponentially localized for all $x\in\Lambda_L$.
\end{enumerate}

\item
Given $m>0$, $\Lambda_L$ will be called $m$-localizing (LOC) for $H$ if the following holds:
\begin{enumerate}
\item $\Lambda_L$ is level spacing for $H_{\Lambda_L}$.
\item  There exists an $m$-localized eigensystem for $H_{\Lambda_L}$.
\end{enumerate}

\end{enumerate}
\end{definition}

\begin{remark}\label{rmk1}
It follows immediately from the definition that given $\widetilde{s}\in(0,1)$,
\beq
\Lambda_L\sqtx{is}m^\ast\text{-mix localizing}\;\Longrightarrow\;\Lambda_L\sqtx{is}\left(1-\frac{\log\frac{40}{m^\ast}}{\log L}\right)\text{-SEL}
\; \Longrightarrow\; \Lambda_L\sqtx{is}\widetilde{s}\text{-SEL},
\eeq
for sufficiently large $L$. (We consider $m^\ast<40$.)
\end{remark}

We now state the bootstrap multiscale analysis. We will use $C_{a,b,\ldots}$, $C'_{a,b,\ldots}$, $C(a,b,\ldots)$, etc., to denote a finite constant depending on the
parameters $a,b,\ldots$. Note that $C_{a,b,\ldots}$ may denote different constants in different equations, and even in the same equation. We will omit the dependence on $d$ and $\mu$ from the notation.

Given $\theta>\left(\tfrac{6}{2\alpha-1}+\tfrac{9}{2}\right)d$ and $0<\xi<1$, we introduce the following parameters:
\begin{itemize}
\item  We fix $q,p,\gamma_1$ such that
\begin{gather}\label{thetaqp}
\tfrac{3d}{2\alpha-1}<q<\tfrac{1}{2}\left(\theta-\tfrac{9}{2}d\right),\quad0<p<(2\alpha-1)q-3d,
\\ \text{and}\quad1<\gamma_1<\min\left\{1+\tfrac{p}{p+2d},\tfrac{2\theta-4d}{5d+4q}\right\},\notag
\end{gather}
and note that
\beq
\theta>2d+\gamma_1\left(\tfrac{5d}{2}+2q\right)>\tfrac{9d}{2}+2q
\eeq
\item We fix $\zeta,\beta,\gamma,\tau$ such that
\begin{gather}\label{zetabetagammatau}
0<\xi<\zeta<\beta< \tfrac 1 \gamma<1<\gamma<\sqrt{\tfrac\zeta\xi},\\    \text{and}\quad
\max\left\{\tfrac{1+\gamma_1}{2\gamma_1},\tfrac{1+\gamma\beta}{2},\tfrac{(\gamma-1)\beta+1}{\gamma}\right\}<\tau<1,\notag
\end{gather}
and note that
\begin{gather}
\tfrac{1}{\gamma_1}<1-\tau+\tfrac{1}{\gamma_1}<\tau,\qtx{and}
\\\nonumber0<\xi<\xi\gamma^2<\zeta<\beta<\tfrac{\tau}{\gamma}<\tfrac{1}{\gamma}<\tau<1<\tfrac{1-\beta}{\tau-\beta}<\gamma<\tfrac{\tau}{\beta}.
\end{gather}
\item We fix $s$ such that
\beq\label{sfix}
\max\left\{\gamma\beta,1-2\gamma\left(\tau-\tfrac{1+\gamma\beta}{2}\right)\right\}<s<1,
\eeq
and note that
\beq
0<\zeta<\beta<\gamma\beta<s<1\qtx{and}1-\tau+\tfrac{1-s}{\gamma}<\tau-\gamma\beta.
\eeq
\item We also let
\beq
\widetilde{\zeta}=\tfrac{\zeta+\beta}{2}\in(\zeta,\beta),\quad\widetilde{\tau}=\tfrac{1+\tau}{2}\in(\tau,1)\qtx{and}L_{\widetilde{\tau}}=\lfloor L^{\widetilde{\tau}}\rfloor.
\eeq
\end{itemize}
In what follows, given   $\theta>\left(\tfrac{6}{2\alpha-1}+\tfrac{9}{2}\right)d$,  we   fix
 $q,p,\gamma_1$ as in \eqref{thetaqp}, and then, given $0<\xi<1$, we fix
 $\zeta,\beta,\gamma,\tau$  as in \eqref{zetabetagammatau}.   We use   Definitions~\ref{defallloc}--\ref{deflocbox} with these fixed   $q,\beta, \tau$, which we omit from the dependence of the constants.

\begin{theorem}\label{mainthm}
Let $\theta>\left(\tfrac{6}{2\alpha-1}+\tfrac{9}{2}\right)d$ and $\varepsilon_0>0$.  There exists a finite scale $\mathcal{L}(\varepsilon_0,\theta)$ with the following property:
Suppose for some  $\eps \in (0,\varepsilon_0]$,  $L_0\geq\mathcal{L}(\varepsilon_0,\theta)$, and $0\leq P_0<\tfrac{1}{2(800)^{2d}}$,  we have
\beq\label{hypmain}
\inf_{x\in\R^d}\P\{\Lambda_{L_0}(x)\sqtx{is}\theta\text{-polynomially localizing for}\;H_{\varepsilon,\omega}\}\geq1-P_0.
\eeq
Then, given $0<\xi<1$, we can find a finite scale $\widetilde{L}=\widetilde{L}(\varepsilon_0,\theta,\xi,L_0)$ and $m_\xi=m(\xi,\widetilde{L})>0$ such that
\beq\label{resmain}
\inf_{x\in\R^d}\P\{\Lambda_L(x)\sqtx{is}m_\xi\text{-localizing for}\;H_{\varepsilon,\omega}\}\geq1-\e^{-L^\xi}\sqtx{for all}L\geq\widetilde{L}.
\eeq
\end{theorem}

The eigensystem bootstrap multiscale analysis, stated in Theorem~\ref{mainthm}, is proven in Section~\ref{sectbmsa}.  It  follows from a repeated use of a bootstrap argument, as in \cite[Section~6]{GKboot}, making successive use of
Propositions~\ref{propmsa1}, \ref{indumsa1one}, \ref{propmsa2}, \ref{propmsa3}, \ref{indumsa3two}, and \ref{propmsa4}. Propositions~\ref{propmsa1}, \ref{propmsa2},
\ref{propmsa3}, and \ref{propmsa4} are eigensystem multiscale analyses. But there is a difference in the procedure  comparing with the Green's function bootstrap multiscale analysis of \cite{GKboot}.
Unlike the definitions of good boxes for the Green's function multiscale analyses, the definitions of good (i.e., localizing) boxes for  the eigensystem multiscale analyses, given in  Definition~\ref{deflocbox},
require intermediate scales, namely  $\tfrac{L}{20}$ and $L^\tau$  in Definition~\ref{defallloc}. For this reason we only have the direct implications given in  Remark~\ref{rmk1}. Thus the
bootstrap between the eigensystem multiscale analyses requires some extra intermediate steps, given in
Propositions~\ref{indumsa1one} and \ref{indumsa3two}.

In Section~\ref{sectinit} we will prove that we can fulfill the hypotheses of Theorem~\ref{mainthm}, obtaining
the following theorem.

\begin{theorem}\label{maincor}
There exists $\varepsilon_0>0$ such that, given $0<\xi<1$, we can find a finite scale $\widetilde{L}=\widetilde{L}(\varepsilon_0,\xi)$ and $m_\xi=m(\xi,\widetilde{L})>0$
such that for all $0<\varepsilon\leq\varepsilon_0$ we have
\beq
\inf_{x\in\R^d}\P\{\Lambda_L(x)\sqtx{is}m_\xi\text{-localizing for}\;H_{\varepsilon,\omega}\}\geq1-\e^{-L^\xi}\sqtx{for all}L\geq\widetilde{L}.
\eeq
\end{theorem}

Theorem~\ref{maincor} yields all the usual forms of localization. To see this, we introduce some notation and definitions. We fix $\nu>\frac{d}{2}$, and set $\langle x\rangle=\sqrt{1+\|x\|^2}$.

A function $\psi:\Z^d\rightarrow\C$ is called a $\nu$-generalized eigenfunction for $H_\varepsilon$ if $\psi$ is a generalized eigenfunction (see \eqref{pteig})
and $0<\|\langle x\rangle^{-\nu}\psi\|<\infty$. We let
$\mathcal{V}_\varepsilon(\lambda)$ denote the collection of $\nu$-generalized eigenfunctions for $H_\varepsilon$ with generalized eigenvalue $\lambda\in\R$.

Given $\lambda\in\R$ and $a,b\in\Z^d$, we set
\beq
W_{\varepsilon,\lambda}^{(a)}(b):=\begin{cases}
\sup_{\psi\in\mathcal{V}_\varepsilon(\lambda)}\frac{|\psi(b)|}{\|\langle x-a\rangle^{-\nu}\psi\|}&
\text{if}\quad\mathcal{V}_\varepsilon(\lambda)\neq\emptyset\\0&\text{otherwise}\end{cases}.
\eeq

Theorem~\ref{maincor} yields
the following theorem, from which one can derive  Anderson localization (pure point spectrum with exponentially decaying eigenfunctions) dynamical localization, and more, as in \cite[Corollary~1.8]{EK}.

\begin{theorem}\label{thmloc}
Let $H_{\varepsilon,\omega}$ be an Anderson model. There exists $\varepsilon_0>0$ such that, given $\xi\in(0,1)$, we can find a scale $\widehat{L}=\widehat{L}(\varepsilon_0,\xi)$ and $m_\xi=m(\xi,\widehat{L})>0$,
such that for all $0<\varepsilon\leq\varepsilon_0$, $L\geq\widehat{L}$ with $L\in2\N$, and $a\in\Z^d$ there exists an event $\mathcal{Y}_{\varepsilon,L,a}$ with the following properties:

\begin{enumerate}
\item $\mathcal{Y}_{\varepsilon,L,a}$ depends only on the random variables $\{\omega_x\}_{x\in\Lambda_{5L}(a)}$, and
\beq
\P\{\mathcal{Y}_{\varepsilon,L,a}\}\geq1-C_{\varepsilon_0}\e^{-L^\xi}.
\eeq
\item For all $\omega\in\mathcal{Y}_{\varepsilon,L,a}$ and $\lambda\in\R$ we have, with
\beq
\max_{b\in\Lambda_{\frac{\ell}{3}}(a)}W^{(a)}_{\varepsilon,\omega,\lambda}(b)>\e^{-\frac{1}{4}{m_\xi}L}\; \Longrightarrow\; \max_{y\in A_L(a)}W^{(a)}_{\varepsilon,\omega,\lambda}(y)\leq\e^{-\frac{7}{132}m_\xi\|y-a\|},
\eeq
where
\beq
A_L(a):=\left\{y\in\Z^d;\tfrac{8}{7}L\leq\|y-a\|\leq\tfrac{33}{14}L\right\}.
\eeq
In particular,
\beq
W^{(a)}_{\varepsilon,\omega,\lambda}(a)W^{(a)}_{\varepsilon,\omega,\lambda}(y)\leq
 \e^{-\frac{7}{132}m_\xi\|y-a\|}\qtx{for all}y\in A_L(a).
\eeq
\end{enumerate}
\end{theorem}

Theorem~\ref{thmloc} is proved in the same way as \cite[Theorem~1.7]{EK}.

\section{Preliminaries to the multiscale analysis}\label{secprel}

We consider a fixed discrete Schr\"odinger operator $H=-\varepsilon\Delta+V$ on $\ell^2(\Z^d)$, where $0<\varepsilon\leq\varepsilon_0$ for a fixed $\varepsilon_0$ and $V$ is a bounded potential.

\subsection{Some basic facts and  definitions}

Let $\Phi\subset\Theta\subset\Z^d$. We define the boundary, exterior boundary, and interior boundary of $\Phi$ relative to $\Theta$, respectively, by
\begin{align}
\boldsymbol{\partial}^\Theta\Phi&=\{(u,v)\in\Phi\times(\Theta\setminus\Phi);|u-v|=1\},
\\\nonumber\partial_{\mathrm{ex}}^\Theta\Phi&=\{v\in(\Theta\setminus\Phi);(u,v)\in \boldsymbol{\partial}^{\Theta}\Phi\qtx{for some}u\in\Phi\},
\\\nonumber\partial^\Theta_{\mathrm{in}}\Phi&=\{u \in {\Phi};(u,v)\in\boldsymbol{\partial}^{\Theta}\Phi\qtx{for some}v\in\Theta\setminus\Phi\}.
\end{align}
We have
\beq
H_\Theta=H_{\Phi}\oplus H_{\Theta\setminus\Phi}+\varepsilon\Gamma_{\boldsymbol{\partial}^\Theta\Phi}\qtx{on}\ell^2(\Theta)=\ell^2(\Phi)\oplus\ell^2( \Theta\setminus\Phi),
\eeq
\beq
\text{where}\quad\Gamma_{\boldsymbol{\partial}^\Theta\Phi}(u,v)=
\begin{cases}
-1&\qtx{if either}(u,v)\sqtx{or}(v,u)\in\boldsymbol{\partial}^\Theta\Phi\\
\ \ 0&\qtx{otherwise}
\end{cases}.
\eeq
For $t\geq1$ we set
\begin{align}
\Phi^{\Theta,t}&=\{y\in\Phi;\Lambda_{2t}(y)\cap\Theta\subset\Phi\}=\{y\in \Phi;\dist(y,\Theta\setminus\Phi)>\lfloor t\rfloor\},
\\\nonumber\partial_{\mathrm{in}}^{\Theta,t}\Phi&=\Phi\setminus\Phi^{\Theta,t}=\{y\in\Phi;\dist(y,\Theta\setminus\Phi)\leq\lfloor t\rfloor\},
\\\nonumber\partial^{\Theta,t}\Phi& =\partial_{\mathrm{in}}^{\Theta,t}\Phi\cup\partial_{\mathrm{ex}}^{\Theta}\Phi.
\end{align}
Given a box $\Lambda_L(x)\subset\Theta\subset\Z^d$ we write $\Lambda_L^{\Theta,t}(x)$ for $(\Lambda_L(x))^{\Theta,t}$.

For a box $\Lambda_L\subset\Theta\subset\Z^d$, there exists a unique $\hat{v}\in\partial_{\mathrm{in}}^{\Lambda_L}\Theta$ for each $v\in\partial_{\mathrm{ex}}^{\Lambda_L}\Theta$ such
that $(\hat{v},v)\in\partial_{\Lambda_L}\Theta$. Given $v\in\Theta$, we define $\hat{v}$ as above if $v\in\partial_{\mathrm{ex}}^{\Lambda_L}\Theta$, and set $\hat{v}=v$ otherwise. Note
that $|\partial_{\mathrm{ex}}^{\Lambda_L}\Theta|=|\boldsymbol{\partial}_{\Lambda_L}\Theta|$. If $L\geq2$, we have
\beq
|\partial_{\mathrm{in}}^\Theta\Lambda_L|\leq|\partial_{\mathrm{ex}}^\Theta\Lambda_L|=|\boldsymbol{\partial}^\Theta\Lambda_L|\leq s_dL^{d-1}, \qtx{where}s_d=2^dd.
\eeq

To cover a box of side $L$ by boxes of side $\ell<L$, we will use suitable covers as in \cite[Definition~3.10]{EK} (also see \cite[Definition~3.12]{GKber}).
\begin{definition}
Let $\Lambda_L=\Lambda_L(x_0)$, $x_0\in\R^d$ be a box in $\Z^d$, and let $\ell<L$. A suitable $\ell$-cover of $\Lambda_L$ is
the collection of boxes
\begin{align}
\mathcal C_{L,\ell}(x_0)=\{\Lambda_\ell(a)\}_{a\in\Xi_{L,\ell}},
\end{align}
where
\beq
\Xi_{L,\ell}:=\{x_0+\rho\ell\Z^{d}\}\cap\Lambda_L^\R\qtx{with}\rho\in[\tfrac{3}{5},\tfrac{4}{5}]\cap\left\{\tfrac{L-\ell}{2\ell k};k\in\N\right\}.
\eeq
We call $\mathcal C_{L,\ell}(x_0)$ the suitable $\ell$-cover of $\Lambda_L$ if
$\rho=\rho_{L,\ell}:=\max\left\{[\tfrac{3}{5},\tfrac{4}{5}]\cap\left\{\tfrac{L-\ell}{2\ell k};k\in\N\right\}\right\}$.
\end{definition}

Note that $[\tfrac{3}{5},\tfrac{4}{5}]\cap\left\{\tfrac{L-\ell}{2\ell k};k\in\N\right\}\neq\emptyset$ if $\ell\leq\tfrac{L}{6}$. For a suitable $\ell$-cover $\mathcal{C}_{L,\ell}(x_0)$,
we have (see \cite[Lemma 3.11]{EK})
\begin{align}
\label{coverprop}\Lambda_L&=\bigcup_{a\in\Xi_{L,\ell}}\Lambda_\ell^{\Lambda_L,\frac{\ell}{10}}(a);
\\\label{covernum}\left(\tfrac{L}{\ell}\right)^{d}&\leq\#\Xi_{L,\ell}=\left(\tfrac{L-\ell}{\rho\ell}+1\right)^d\leq \left(\tfrac{2L}{\ell}\right)^d.
\end{align}

\subsection{Lemmas about eigenpairs}\label{sectlem}

Given $\Theta\subset\Z^d$ and an eigensystem $\{(\varphi_j,\lambda_j)\}_{j\in J}$ for $H_\Theta$. We have
\begin{align}\label{eigsys}
\delta_y&=\sum_{j\in J}\overline{\varphi_j(y)}\varphi_j\qtx{for all}y\in\Theta,
\\\nonumber\psi(y)&=\langle\delta_y,\psi\rangle=\sum_{j\in J}\varphi_j(y)\langle\varphi_j,\psi\rangle\qtx{for all}\psi\in\ell^2(\Theta)\qtx{and}y\in \Theta.
\end{align}

Given $\Theta\subset\Z^d$, a function $\psi:\Theta\rightarrow\C$ is called a generalized eigenfunction for $H_\Theta$ with generalized eigenvalue $\lambda\in\R$
if $\psi$ is not identically zero and
\beq
-\varepsilon\hspace{-10pt}\sum_{\  y\in\Theta,|y-x|=1}\psi(y)+(V(x)-\lambda)\psi(x)=0\qtx{for all}x\in{\Theta},
\eeq
or, equivalently,
\beq\label{pteig}
\langle(H_\Theta-\lambda)\varphi,\psi\rangle=0\qtx{for all}\varphi\in\ell^2(\Theta)\quad\text{with finite support}.
\eeq
If $\psi\in\ell^2(\Theta)$, $\psi$ is an eigenfunction for $H_{\Theta}$ with eigenvalue $\lambda$. We do not require generalized eigenfunctions to be in $\ell^2(\Theta)$,
we only require the pointwise equality in \eqref{pteig}. If $\Theta$ is finite there is no difference between generalized eigenfunctions and eigenfunctions.

\begin{lemma}\label{distdecayall}
Consider a box $\Lambda_L\subset\Theta\subset\Z^d$, and suppose $(\varphi,\lambda)$ is an eigenpair for
$H_{\Lambda_L}$. Then:

\begin{enumerate}
\item Given $\widetilde{\theta}>0$, if $\varphi$ is $(x,\widetilde{\theta})$-polynomially localized for some $x\in\Lambda_L^{\Theta,L'}$, we have
\beq
\dist(\lambda,\sigma(H_\Theta))\leq\|(H_\Theta-\lambda)\varphi\|\leq C_{d,\varepsilon_0}L^{-\left(\widetilde{\theta}-\frac{d-1}{2}\right)}.
\eeq

\item Given $\widetilde{s}\in(0,1)$, if $\varphi$ is $(x,\widetilde{s})$-subexponentially localized for some $x\in\Lambda_L^{\Theta,L'}$, we have
\begin{gather}
\dist(\lambda,\sigma(H_\Theta))\leq\|(H_\Theta-\lambda)\varphi\|\leq\e^{-c_1L^{\widetilde{s}}},
\\\label{sc1}\text{where}\quad
c_1=c_1(L)\geq1-C_{d,\varepsilon_0}\tfrac{\log L}{L^{\widetilde{s}}}.
\end{gather}

\item Given $m>0$ and $\tau\in(0,1)$, if $\varphi$ is $(x,m)$ localized for some $x\in\Lambda_L^{\Theta,L_\tau}$, we have
\begin{gather}
\dist(\lambda,\sigma(H_\Theta))\leq\|(H_\Theta-\lambda)\varphi\|\leq\e^{-m_1L_\tau},
\\\label{mst1}\text{where}\quad
m_1=m_1(L)\geq m-C_{d,\varepsilon_0}\tfrac{\log L}{L_\tau}.
\end{gather}
\end{enumerate}

\end{lemma}

\begin{proof}
We prove part (i), the proofs of (ii) and (iii) are similar. If $x\in\Lambda_L^{\Theta,L'}$, we have $\dist(x,\partial_{\mathrm{in}}^\Theta\Lambda_L)\geq L'$,
thus it follows from \cite[Lemma~3.2]{EK} that
\begin{align}
\|(H_\Theta-\lambda)\varphi\|&\leq\varepsilon\sqrt{s_d}L^{\frac{d-1}{2}}\|\varphi_{\partial_{\mathrm{in}}^\Theta\Lambda_L}\|_\infty\leq\varepsilon\sqrt{s_d}L^{\frac{d-1}{2}}L^{-\widetilde{\theta}}
\\\nonumber&\leq\varepsilon_0\sqrt{s_d}L^{-\left(\widetilde{\theta}-\frac{d-1}{2}\right)}.
\end{align}
\end{proof}

For the following lemmas in this and next subsections, we fix $\theta>\left(\tfrac{6}{2\alpha-1}+\tfrac{9}{2}\right)d$ and $0<\xi<1$
(so $q,p,\gamma_1,\zeta,\beta,\gamma,\tau,s$ are fixed). Also, when we consider $\Lambda_\ell$ to be a $\sharp$ box, where $\sharp$ stands for $\theta$-PL, $m^\ast$-ML,
$s$-SEL or $m$-LOC, with $m^\ast\geq m^\ast_-(\ell)>0$ and $m\geq m_-(\ell)>0$, we let:
\beq\label{lsharp}
L=L_\sharp=\begin{cases}
Y\ell\sqtx{or}\ell^{\gamma_1}&\text{if}\;\sharp\sqtx{is}\theta\text{-PL}\\
\ell^{\gamma_1}&\text{if}\;\sharp\sqtx{is}m^\ast\text{-ML}\\
Y\ell\sqtx{or}\ell^\gamma&\text{if}\;\sharp\sqtx{is}s\text{-SEL}\\
\ell^\gamma&\text{if}\;\sharp\sqtx{is}m\text{-LOC}
\end{cases}
\qtx{and}
\ell_\sharp=\begin{cases}
\ell'&\text{if}\;\sharp\sqtx{is}\theta\text{-PL}\sqtx{or}s\text{-SEL}\\
\ell_\tau&\text{if}\;\sharp\sqtx{is}m^\ast\text{-ML}\sqtx{or}m\text{-LOC}
\end{cases},
\eeq
where $Y\geq1$. We will omit the dependence on $\theta$, $\xi$ and $Y$ from the notation.

We prove most of the lemmas only for $\sharp$ being $\theta$-PL. The proofs of other cases are similar.

\begin{lemma}\label{psidecayall}
Given $\Theta\subset\Z^d$, let $\psi:\Theta\rightarrow\C$ be a generalized eigenfunction for $H_\Theta$ with generalized eigenvalue $\lambda\in\R$.
Consider a $\sharp$ box $\Lambda_\ell\subset\Theta$ with a corresponding eigensystem $\{(\varphi_u,\nu_u)\}_{u\in\Lambda_\ell}$, and suppose for all
$u\in\Lambda_\ell^{\Theta,\ell_\sharp}$ we have
\beq\label{disteig}
|\lambda-\nu_u|\geq\begin{cases}
\frac{1}{2}L^{-q}&\text{if}\;\sharp\sqtx{is}\theta\text{-PL}\sqtx{or}m^\ast\text{-ML}\\
\frac{1}{2}\e^{-L^\beta}&\text{if}\;\sharp\sqtx{is}s\text{-SEL}\sqtx{or}m\text{-LOC}
\end{cases}.
\eeq
Then the following holds for sufficiently large $\ell$:
\begin{enumerate}
\item Let $y\in\Lambda_\ell^{\Theta,2\ell_\sharp}$. Then:
\begin{enumerate}
\item If $\sharp$ is $\theta$-PL, we have
\beq\label{decayes1}
|\psi(y)|\leq C_{d,\varepsilon_0}L^q\ell^{-(\theta-2d)}|\psi(y_1)|\qtx{for some}y_1\in\partial^{\Theta,2\ell'}\Lambda_\ell.
\eeq
\item If $\sharp$ is $s$-SEL, we have
\begin{gather}\label{decayes1sub}
|\psi(y)|\leq\e^{-c_2\ell^s}|\psi(y_1)|\qtx{for some}y_1\in\partial^{\Theta,2\ell'}\Lambda_\ell,
\\\label{sc2}\text{where}\quad
c_2=c_2(\ell)\geq1-C_{d,\varepsilon_0}L^\beta\ell^{-s}.
\end{gather}
\item If $\sharp$ is $m^\ast$-ML, we have
\begin{gather}\label{decayes1str}
|\psi(y)|\leq \e^{-m^\ast_2\ell_\tau}|\psi(y_1)|\qtx{for some}y_1\in\partial^{\Theta,2\ell_\tau}\Lambda_\ell,
\\\label{mst2}\text{where}\quad
m^\ast_2=m^\ast_2(\ell)\geq m^\ast-C_{d,\varepsilon_0}\gamma_1q\tfrac{\log\ell}{\ell_\tau}.
\end{gather}
\item If $\sharp$ is $m$-LOC, we have
\begin{gather}\label{decayes1loc}
|\psi(y)|\leq \e^{-m_2\ell_\tau}|\psi(y_1)|\qtx{for some}y_1\in\partial^{\Theta,2\ell_\tau}\Lambda_\ell,
\\\label{m2}\text{where}\quad
m_2=m_2(\ell)\geq m-C_{d,\varepsilon_0}\ell^{\gamma\beta-\tau}.
\end{gather}
\end{enumerate}
\item Let $y\in\Lambda_\ell^{\Theta,2\ell_{\widetilde{\tau}}}$. Then:
\begin{enumerate}
\item If $\sharp$ is $m^\ast$-ML, we have
\begin{gather}\label{decayes2str}
|\psi(y)|\leq\e^{-m^\ast_3\|y_2-y\|}|\psi(y_2)|\qtx{for some}y_2\in\partial^{\Theta,\ell_{\widetilde{\tau}}}\Lambda_\ell,
\\\label{mst3}\text{where}\quad
m^\ast_3=m^\ast_3(\ell)\geq m^\ast\left(1-4\ell^{\frac{\tau-1}{2}}\right)-C_{d,\varepsilon_0}\gamma_1q\tfrac{\log\ell}{\ell_{\widetilde{\tau}}}.
\end{gather}
\item If $\sharp$ is $m$-LOC, we have
\begin{gather}\label{decayes2loc}
|\psi(y)|\leq \e^{-m_3\|y_2-y\|}|\psi(y_2)|\qtx{for some}y_2\in\partial^{\Theta,\ell_{\widetilde{\tau}}}\Lambda_\ell,
\\\label{m3}\text{where}\quad
m_3=m_3(\ell)\geq m\left(1-4\ell^{\frac{\tau-1}{2}}\right)-C_{d,\varepsilon_0}\ell^{\gamma\beta-\widetilde{\tau}}.
\end{gather}
\end{enumerate}
\end{enumerate}
\end{lemma}

\begin{proof}
Let $y\in\Lambda_\ell$, we have (see \eqref{eigsys})
\beq\label{sum1}
\psi(y)=\sum_{u\in\Lambda_\ell}\varphi_u(y)\langle\varphi_u,\psi\rangle=\sum_{u\in\Lambda_\ell^{\Theta,\ell'}}\varphi_u(y)\langle\varphi_u,\psi\rangle
+\sum_{u\in\partial_{\mathrm{in}}^{\Theta,\ell'}\Lambda_\ell}\varphi_u(y)\langle\varphi_u,\psi\rangle.
\eeq

If $u\in\Lambda_\ell^{\Theta,\ell'}$, we have $|\lambda-\nu_u|\geq\frac{1}{2}L^{-q}$ by \eqref{disteig}. Using \eqref{pteig}, we get
\beq
\langle\varphi_u,\psi\rangle=(\lambda-\nu_u)^{-1}\langle\varphi_u,(H_\Theta-\nu_u)\psi\rangle=(\lambda-\nu_u)^{-1}\langle(H_\Theta-\nu_u)\varphi_u,\psi\rangle.
\eeq
It follows from \cite[Lemma~3.2]{EK} that
\beq\label{vphipsi1}
|\varphi_u(y)\langle\varphi_u,\psi\rangle|\leq2L^q\varepsilon\sum_{v\in\partial_{\mathrm{ex}}^\Theta\Lambda_\ell}|\varphi_u(y)\varphi_u(\hat{v})||\psi(v)|.
\eeq
If $v'\in\partial_{\mathrm{in}}^\Theta\Lambda_\ell$, we have $\|v'-u\|\geq\ell'$, so \eqref{defploc} gives $|\varphi_u(v')|\leq\ell^{-\theta}$.
It follows from \eqref{vphipsi1} and $\|\varphi_u\|=1$ that
\beq
|\varphi_u(y)\langle\varphi_u,\psi\rangle|\leq2\varepsilon L^q\ell^{-\theta}\sum_{v\in\partial_{\mathrm{ex}}^\Theta\Lambda_\ell}|\psi(v)|\leq2\varepsilon s_dL^q\ell^{-(\theta-d+1)}|\psi(v_1)|
\eeq
for some $v_1\in\partial_{\mathrm{ex}}^\Theta\Lambda_\ell$. Therefore
\beq\label{sum2}
\left|\sum_{u\in\Lambda_\ell^{\Theta,\ell'}}\varphi_u(y)\langle\varphi_u,\psi\rangle\right|\leq2\varepsilon s_d L^q\ell^{-(\theta-2d+1)}|\psi(v_2)|
\eeq
for some $v_2\in\partial_{\mathrm{ex}}^\Theta\Lambda_\ell$.

Let $y\in\Lambda_\ell^{\Theta,2\ell'}$. If $u\in\partial_{\mathrm{in}}^{\Theta,\ell'}\Lambda_\ell$, we have $\|u-y\|\geq2\ell'-\ell'=\ell'$, thus \eqref{defploc}
gives $|\varphi_u(y)|\leq\ell^{-\theta}$, and hence
\beq\label{sum3}
\left|\sum_{u\in\partial_{\mathrm{in}}^{\Theta,\ell'}\Lambda_\ell}\varphi_u(y)\langle\varphi_u,\psi\rangle\right|
\leq\ell^{-(\theta-d)}\|\psi\chi_{\Lambda_\ell}\|\leq\ell^{-\left(\theta-\frac{3d}{2}\right)}|\psi(v_3)|
\eeq
for some $v_3\in\Lambda_\ell$. Combining \eqref{sum1}, \eqref{sum2} and \eqref{sum3}, we conclude that
\beq\label{ca1}
|\psi(y)|\leq(1+2\varepsilon_0s_d)L^q\ell^{-(\theta-2d)}|\psi(y_1)|
\eeq
for some $y_1\in\Lambda_\ell\cup\partial_{\mathrm{ex}}^\Theta\Lambda_\ell$. If $y_1\not\in\partial^{\Theta,2\ell'}\Lambda_\ell$ we repeat the procedure to estimate $|\psi(y_1)|$.
Since we can suppose $\psi(y)\neq0$ without loss of generality, the procedure must stop after finitely many times, and at that time we must have \eqref{decayes1}.

We prove part (ii) only for $\sharp$ being $m^\ast$-ML. The proof for $\sharp$ being $m$-LOC is similar. Let $y\in\Lambda_\ell^{\Theta,\ell_{\widetilde{\tau}}}$,
then $\|y-v'\|\geq\ell_{\widetilde{\tau}}$ for $v'\in\partial_{\mathrm{in}}^\Theta\Lambda_\ell$. Thus for $u\in\Lambda_\ell^{\Theta,\ell_\tau}$
and $v'\in\partial_{\mathrm{in}}^\Theta\Lambda_\ell$ we have
\beq\label{vphipsi3str}
|\varphi_u(y)\varphi_u(v')|\leq
\begin{cases}
\e^{-m^\ast(\|y-u\|+\|v'-u\|)}\leq\e^{-m^\ast\|v'-y\|}&\qtx{if}\|y-u\|\geq\ell_\tau
\\\e^{-m^\ast\|v'-u\|}\leq\e^{-m'_1\|v'-y\|}&\qtx{if}\|y-u\|<\ell_\tau
\end{cases},
\eeq
where
\beq
m'_1\geq m^\ast\left(1-2\ell^{\tau-\widetilde{\tau}}\right)=m^\ast\left(1-2\ell^{\frac{\tau-1}{2}}\right),
\eeq
since for $\|y-u\|<{\ell_\tau}$, we have
\beq
\|v'-u\|\geq\|v'-y\|-\|y-u\|\geq\|v'-y\|-\ell_\tau\geq\|v'-y\|\left(1-\tfrac{\ell_\tau}{\ell_{\widetilde{\tau}}}\right).
\eeq
Combining \eqref{vphipsi1} and \eqref{vphipsi3str}, we conclude that
\begin{align}\label{vphipsi4str}
&|\varphi_u(y)\langle\varphi_u,\psi\rangle|\leq2\varepsilon L^q\sum_{v\in\partial_{\mathrm{ex}}^\Theta\Lambda_\ell}\e^{-m'_1(\|v -y\|-1)}|\psi(v)|
\\&\nonumber\qquad\leq2\varepsilon s_d\ell^{\gamma_1q+d-1}\e^{-m'_1(\|v_1-y\|-1)}|\psi(v_1)|\leq\e^{-m'_2\|v_1-y\|}|\psi(v_1)|
\end{align}
for some $v_1\in\partial_{\mathrm{ex}}^\Theta\Lambda_\ell$, where we used $\|v_1-y\|\geq\ell_{\widetilde{\tau}}$ and took
\beq
m'_2\geq m'_1\left(1-2\ell^{\widetilde{\tau}}\right)-C_{d,\varepsilon_0}\gamma_1q\tfrac{\log\ell}{\ell_{\widetilde{\tau}}}
\geq m^\ast\left(1-4\ell^{\frac{\tau-1}{2}}\right)-C_{d,\varepsilon_0}\gamma_1q\tfrac{\log\ell}{\ell_{\widetilde{\tau}}}.
\eeq
Therefore
\beq\label{sum4str}
\left|\sum_{u\in\Lambda_\ell^{\Theta,\ell_\tau}}\varphi_u(y)\langle\varphi_u,\psi\rangle\right|\leq\ell^d\e^{-m'_2\|v_2-y\|}|\psi(v_2)|\leq \e^{-m'_3\|v_2-y\|}|\psi(v_2)|
\eeq
for some $v_2\in\partial_{\mathrm{ex}}^\Theta\Lambda_\ell$, where
\beq
m'_3\geq m'_2-C_d\tfrac{\log\ell}{\ell_{\widetilde{\tau}}}\geq m^\ast\left(1-4\ell^{\frac{\tau-1}{2}}\right)-C_{d,\varepsilon_0}\gamma_1q\tfrac{\log\ell}{\ell_{\widetilde{\tau}}}.
\eeq

If $u\in\partial_{\mathrm{in}}^{\Theta,\ell_\tau}\Lambda_\ell$ we have
$\|u-y\|\geq\ell_{\widetilde{\tau}}-\ell_\tau>\tfrac{1}{2}\ell_{\widetilde{\tau}}$, thus \eqref{defloc} gives $|\varphi_u(y)|\leq\e^{-m^\ast\|u-y\|}$. Also, \eqref{defloc} implies
\beq
|\varphi_u(v)|\leq\e^{m^\ast\ell_\tau}\e^{-m^\ast\|v-u\|}\qtx{for all}v\in\Lambda_\ell.
\eeq
Therefore
\beq
|\langle\varphi_u,\psi\rangle|=\left|\sum_{v\in\Lambda_\ell}{\varphi_u(v)}\psi(v)\right|\leq\sum_{v\in\Lambda_\ell}\e^{-m^\ast(\|v-u\|-\ell_\tau)}|\psi(v)|,
\eeq
so we get
\begin{align}
&|\varphi_u(y)\langle\varphi_u,\psi\rangle|\leq\sum_{v\in\Lambda_\ell}\e^{-m^\ast(\|u-y\|-\ell_\tau+\|v-u\|)}|\psi(v)|
\\\nonumber&\quad\leq(\ell+1)^d\e^{-m^\ast(\|u-y\|-\ell_\tau)-m^\ast\|v_3-u\|}|\psi(v_3)|
\\\nonumber&\quad\leq\e^{-m'_4\|u-y\|-m^\ast\|v_3-u\|}|\psi(v_3)|
\\\nonumber&\quad\leq\e^{-m'_4\max\{\|v_3-y\|,\|u-y\|\}}|\psi(v_3)|\leq\e^{-m'_4\max\{\|v_3-y\|,\frac{1}{2}\ell_{\widetilde{\tau}}\}}|\psi(v_3)|
\end{align}
for some $ v_3\in\Lambda_\ell$, where we used $\|u-y\|\geq\tfrac{1}{2}\ell_{\widetilde{\tau}}$ and took
\beq
m'_4\geq m^\ast\left(1-4\ell^{\frac{\tau-1}{2}}\right)-C_d\tfrac{\log\ell}{\ell_{\widetilde{\tau}}}.
\eeq
Therefore
\begin{align}\label{sum5str}
\left|\sum_{u\in\partial_{\mathrm{in}}^{\Theta,\ell_\tau}\Lambda_\ell}\varphi_u(y)\langle\varphi_u,\psi\rangle\right|
&\leq\ell^d\e^{-m'_4\max\{\|v_3-y\|,\frac{1}{2}\ell_{\widetilde{\tau}}\}}|\psi(v_3)|
\\\nonumber&\leq\e^{-m'_5\max\{\|v_3-y\|,\frac{1}{2}\ell_{\widetilde{\tau}}\}}|\psi(v_3)|
\end{align}
for some $v_3\in\Lambda_\ell$, where
\beq
m'_5\geq{m^\ast_4}'-C_d\tfrac{\log\ell}{\ell_{\widetilde{\tau}}}\geq m^\ast\left(1-4\ell^{\frac{\tau-1}{2}}\right)-C_d\tfrac{\log\ell}{\ell_{\widetilde{\tau}}}.
\eeq
Combining \eqref{sum1}, \eqref{sum4str}, and \eqref{sum5str}, we conclude that
\beq\label{decayest1str}
|\psi(y)|\leq\e^{-m^\ast_3\max\{\|y_1-y\|,\frac{1}{2}\ell_{\widetilde{\tau}}\}}|\psi(y_1)|\qtx{for some}y_1\in\Lambda_\ell\cup\partial_{\mathrm{ex}}^\Theta\Lambda_\ell,
\eeq
where $m^\ast_3$ is given in \eqref{mst3}. If $y_1\not\in\partial^{\Theta,\ell_{\widetilde{\tau}}}\Lambda_\ell$ we repeat the procedure to estimate $|\psi(y_1)|$.
Since we can suppose $\psi(y)\neq0$ without loss of generality, the procedure must stop after finitely many times, and at that time we must have
\beq\label{decayest2str}
|\psi(y)|\leq\e^{-m^\ast_3\max\{\|\widetilde{y}-y\|,\frac{1}{2}\ell_{\widetilde{\tau}}\}}|\psi(\widetilde{y})|
\qtx{for some}\widetilde{y}\in\partial^{\Theta,\ell_{\widetilde{\tau}}}\Lambda_\ell.
\eeq

If $y\in\Lambda_\ell^{\Theta,2\ell_{\widetilde{\tau}}}$, \eqref{decayes2str} follows immediately from \eqref{decayest2str}.
\end{proof}

\begin{lemma}\label{ideigsysall}
Given a finite set $\Theta\subset\Z^d$, let $\{(\psi_\lambda,\lambda)\}_{\lambda\in\sigma(H_\Theta)}$ be an eigensystem for $H_\Theta$.

Then the following holds for sufficiently large $\ell$:
\begin{enumerate}
\item Let $\Lambda_\ell(a)\subset\Theta$, where $a\in\R^d$, be a $\sharp$-localizing box with a corresponding eigensystem
$\left\{(\varphi_x^{(a)},\lambda_x^{(a)})\right\}_{x\in\Lambda_\ell(a)}$, and let $\Theta$ be $L$-polynomially level spacing
for $H$ if $\sharp$ is $\theta$-PL or $m^\ast$-ML, $L$-level spacing for $H$ if $\sharp$ is $s$-SEL or $m$-LOC.
    \begin{enumerate}
    \item There exists an injection
    \beq
    x\in\Lambda_\ell^{\Theta,\ell_\sharp}(a)\mapsto\widetilde{\lambda}_x^{(a)}\in\sigma(H_\Theta),
    \eeq
    such that for all $x\in\Lambda_\ell^{\Theta,\ell_\sharp}(a)$:
    \begin{enumerate}
    \item If $\sharp$ is $\theta$-PL, we have
    \beq\label{tildedist}
    \left|\widetilde{\lambda}_x^{(a)}-\lambda_x^{(a)}\right|\leq C_{d,\varepsilon_0}\ell^{-\left(\theta-\frac{d-1}{2}\right)},
    \eeq
    and, multiplying each $\varphi_x^{(a)}$ by a suitable phase factor,
    \beq\label{diffun}
    \left\|\psi_{\widetilde{\lambda}_x^{(a)}}-\varphi_x^{(a)}\right\|\leq2C_{d,\varepsilon_0}L^q\ell^{-\left(\theta-\frac{d-1}{2}\right)}.
    \eeq
    \item If $\sharp$ is $s$-SEL, we have
    \beq\label{tildedistsub}
    \left|\widetilde{\lambda}_x^{(a)}-\lambda_x^{(a)}\right|\leq\e^{-c_1\ell^s},\sqtx{with}c_1=c_1(\ell)\sqtx{as in \eqref{sc1},}
    \eeq
    and, multiplying each $\varphi_x^{(a)}$ by a suitable phase factor,
    \beq\label{diffunsub}
    \left\|\psi_{\widetilde{\lambda}_x^{(a)}}-\varphi_x^{(a)}\right\|\leq2\e^{-c_1\ell^s}\e^{L^\beta}.
    \eeq
    \item If $\sharp$ is $m^\ast$-ML, we have
    \beq\label{tildediststr}
    \left|\widetilde{\lambda}_x^{(a)}-\lambda_x^{(a)}\right|\leq\e^{-m^\ast_1\ell_\tau},\sqtx{with}m^\ast_1=m^\ast_1(\ell)\sqtx{as in \eqref{mst1},}
    \eeq
    and, multiplying each $\varphi_x^{(a)}$ by a suitable phase factor,
    \beq\label{diffunstr}
    \left\|\psi_{\widetilde{\lambda}_x^{(a)}}-\varphi_x^{(a)}\right\|\leq2\e^{-m^\ast_1\ell_\tau}L^q.
    \eeq
    \item If $\sharp$ is $m$-LOC, we have
    \beq
    \left|\widetilde{\lambda}_x^{(a)}-\lambda_x^{(a)}\right|\leq\e^{-m_1\ell_\tau},\sqtx{with}m_1=m_1(\ell)\sqtx{as in \eqref{mst1},}
    \eeq
    and, multiplying each $\varphi_x^{(a)}$ by a suitable phase factor,
    \beq\label{diffunloc}
    \left\|\psi_{\widetilde{\lambda}_x^{(a)}}-\varphi_x^{(a)}\right\|\leq2\e^{-m_1\ell_\tau}\e^{L^\beta}.
    \eeq
    \end{enumerate}
    \item Set
    \beq
    \sigma_{\{a\}}(H_\Theta):=\left\{\widetilde{\lambda}_x^{(a)};x\in\Lambda_\ell^{\Theta,\ell_\sharp}(a)\right\}.
    \eeq
    Then if $\lambda\in\sigma_{\{a\}}(H_\Theta)$, for all $y\in\Theta\setminus\Lambda_\ell(a)$ we have
    \beq\label{psidec0}
    |\psi_\lambda(y)|\leq\begin{cases}
    2C_{d,\varepsilon_0}L^q\ell^{-\left(\theta-\frac{d-1}{2}\right)}&\text{if}\;\sharp\sqtx{is}\theta\text{-PL}\\
    2\e^{-c_1\ell^s}\e^{L^\beta}&\text{if}\;\sharp\sqtx{is}s\text{-SEL}\\
    2\e^{-m^\ast_1\ell_\tau}L^q&\text{if}\;\sharp\sqtx{is}m^\ast\text{-ML}\\
    2\e^{-m_1\ell_\tau}\e^{L^\beta}&\text{if}\;\sharp\sqtx{is}m\text{-LOC}
    \end{cases}.
    \eeq
    \item If $\lambda\in\sigma(H_\Theta)\setminus\sigma_{\{a\}}(H_\Theta)$, for all $x\in\Lambda_\ell^{\Theta,\ell_\sharp}(a)$ we have
    \beq\label{distlambda}
    \left|\lambda-\lambda_x^{(a)}\right|\geq\begin{cases}
    \tfrac{1}{2}L^{-q}&\text{if}\;\sharp\sqtx{is}\theta\text{-PL}\sqtx{or}m^\ast\text{-ML}\\
    \tfrac{1}{2}\e^{-L^\beta}&\text{if}\;\sharp\sqtx{is}s\text{-SEL}\sqtx{or}m\text{-LOC}
    \end{cases},
    \eeq
    and for all $y\in\Lambda_\ell^{\Theta,2\ell_\sharp}(a)$,
   \beq\label{psidec1}
    |\psi_\lambda(y)|\leq\begin{cases}
    C_{d,\varepsilon_0}L^q\ell^{-(\theta-2d)}|\psi_\lambda(y_1)|&\text{if}\;\sharp\sqtx{is}\theta\text{-PL}\\
    \e^{-c_2\ell^s}|\psi_\lambda(y_1)|&\text{if}\;\sharp\sqtx{is}s\text{-SEL}\\
    \e^{-m^\ast_2\ell_\tau}|\psi_\lambda(y_1)|&\text{if}\;\sharp\sqtx{is}m^\ast\text{-ML}\\
    \e^{-m_2\ell_\tau}|\psi_\lambda(y_1)|&\text{if}\;\sharp\sqtx{is}m\text{-LOC}
    \end{cases}
    \eeq
    for some $y_1\in\partial^{\Theta,2\ell_\sharp}\Lambda_\ell(a)$, where $c_2=c_2(\ell)$ as in \eqref{sc2}, $m^\ast_2=m^\ast_2(\ell)$ as in \eqref{mst2}, $m_2=m_2(\ell)$ as in \eqref{m2}. Moreover, for all $y\in\Lambda_\ell^{\Theta,2\ell_{\widetilde{\tau}}}(a)$,
    \beq\label{pdec01}
    |\psi_\lambda(y)|\leq\begin{cases}
    \e^{-m^\ast_3\|y_2-y\|}|\psi_\lambda(y_2)|&\text{if}\;\sharp\sqtx{is}m^\ast\text{-ML}\\
    \e^{-m_3\|y_2-y\|}|\psi_\lambda(y_2)|&\text{if}\;\sharp\sqtx{is}m\text{-LOC}
    \end{cases}
    \eeq
    for some $y_2\in\partial^{\Theta,\ell_{\widetilde{\tau}}}\Lambda_\ell(a)$, where $m^\ast_3=m^\ast_3(\ell)$ as in \eqref{mst3}, $m_3=m_3(\ell)$ as in \eqref{m3}.
    \end{enumerate}
\item Let $\{\Lambda_\ell(a)\}_{a\in\mathcal{G}}$, where $\mathcal{G}\subset\R^d$ such that $\Lambda_\ell(a)\subset\Theta$ for all $a\in\mathcal{G}$, be a collection of $\sharp$
boxes with corresponding eigensystems $\left\{(\varphi_x^{(a)},\lambda_x^{(a)})\right\}_{x\in\Lambda_\ell(a)}$ and let $\Theta$ be $L$-polynomially level spacing for $H$
if $\sharp$ is $\theta$-PL or $m^\ast$-ML,
$L$-level spacing for $H$ if $\sharp$ is $s$-SEL or $m$-LOC. Set
    \begin{align}\label{defE}
    \mathcal{E}_{\mathcal{G}}^\Theta(\lambda)&=\left\{\lambda_x^{(a)};a\in\mathcal{G},x\in\Lambda_\ell^{\Theta,\ell_\sharp}(a),
    \widetilde{\lambda}_x^{(a)}=\lambda\right\}\sqtx{for}\lambda\in\sigma(H_\Theta),
    \\\nonumber\sigma_{\mathcal{G}}(H_\Theta)&=\left\{\lambda\in\sigma(H_\Theta);\mathcal{E}_{\mathcal{G}}^\Theta(\lambda)\neq\emptyset\right\}=\bigcup_{;a\in\mathcal{G}}\sigma_{\{a\}}(H_\Theta).
    \end{align}
    \begin{enumerate}
    \item For $a,b\in\mathcal{G}$, $a\neq b$, if $x\in\Lambda_\ell^{\Theta,\ell_\sharp}(a)$ and $y\in\Lambda_\ell^{\Theta,\ell_\sharp}(b)$,
    \beq\label{difxy}
    \lambda_x^{(a)},\lambda_x^{(b)}\in\mathcal{E}_{\mathcal{G}}^\Theta(\lambda)\Longrightarrow\|x-y\|<2\ell_\sharp.
    \eeq
    As a consequence,
    \beq\label{sigmaab}
    \Lambda_\ell(a)\cap\Lambda_\ell(b)=\emptyset\Longrightarrow\sigma_{\{a\}}(H_\Theta)\cap\sigma_{\{b\}}(H_\Theta)=\emptyset.
    \eeq
    \item If $\lambda\in\sigma_{\mathcal{G}}(H_\Theta)$, we have for all $y\in\Theta\setminus\Theta_{\mathcal{G}}$, where $\Theta_{\mathcal{G}}:=\bigcup_{a\in\mathcal{G}}\Lambda_\ell(a)$,
    \beq\label{psidec3}
    |\psi_\lambda(y)|\leq\begin{cases}
    2C_{d,\varepsilon_0}L^q\ell^{-\left(\theta-\frac{d-1}{2}\right)}&\text{if}\;\sharp\sqtx{is}\theta\text{-PL}\\
    2\e^{-c_1\ell^s}\e^{L^\beta}&\text{if}\;\sharp\sqtx{is}s\text{-SEL}\\
    2\e^{-m^\ast_1\ell_\tau}L^q&\text{if}\;\sharp\sqtx{is}m^\ast\text{-ML}\\
    2\e^{-m_1\ell_\tau}\e^{L^\beta}&\text{if}\;\sharp\sqtx{is}m\text{-LOC}
    \end{cases}.
    \eeq
    \item If $\lambda\in\sigma(H_\Theta)\setminus\sigma_{\mathcal{G}}(H_\Theta)$, we have for all $y\in\Theta'_{\mathcal{G}}:=\bigcup_{a\in\mathcal{G}}\Lambda_\ell^{\Theta,2\ell_\sharp}(a)$,
    \beq\label{psidec4}
    |\psi_\lambda(y)|\leq\begin{cases}
    C_{d,\varepsilon_0}L^q\ell^{-(\theta-2d)}&\text{if}\;\sharp\sqtx{is}\theta\text{-PL}\\
    \e^{-c_2\ell^s}&\text{if}\;\sharp\sqtx{is}s\text{-SEL}\\
    \e^{-m^\ast_2\ell_\tau}&\text{if}\;\sharp\sqtx{is}m^\ast\text{-ML}\\
    \e^{-m_2\ell_\tau}&\text{if}\;\sharp\sqtx{is}m\text{-LOC}
    \end{cases}.
    \eeq
    \item If $|\Theta|\leq(L+1)^d$, we have
    \beq
    |\Theta'_{\mathcal{G}}|\leq|\sigma_{\mathcal{G}}(H_\Theta)|\leq|\Theta_{\mathcal{G}}|.
    \eeq
    \end{enumerate}
\end{enumerate}
\end{lemma}

\begin{proof}
Let $\Lambda_\ell(a)\subset\Theta$, where $a\in\R^d$, be a $\theta$-polynomially localizing box with a corresponding eigensystem $\left\{(\varphi_x^{(a)},\lambda_x^{(a)})\right\}_{x\in\Lambda_\ell(a)}$.
It follows from Lemma~\ref{distdecayall} that there exists $\widetilde{\lambda}_x^{(a)}\in\sigma(H_\Theta)$ satisfying \eqref{tildedist} for $x\in\Lambda_\ell^{\Theta,\ell'}(a)$. $\widetilde{\lambda}_x^{(a)}$
is unique since $\Theta$ is $L$-polynomially level spacing for $H_\Theta$ and $q<\gamma_1q<\theta-\tfrac{d-1}{2}$. Moreover, we have $\widetilde{\lambda}_x^{(a)}\neq\widetilde{\lambda}_y^{(a)}$
if $x,y\in\Lambda_\ell^{\Theta,\ell'}(a)$, $x\neq y$, since
\begin{align}
\left|\widetilde{\lambda}_x^{(a)}-\widetilde{\lambda}_y^{(a)}\right|&\geq\left|\lambda_x^{(a)}-\lambda_y^{(a)}\right|-\left|\widetilde{\lambda}_x^{(a)}-\lambda_x^{(a)}\right|
-\left|\widetilde{\lambda}_y^{(a)}-\lambda_y^{(a)}\right|
\\\nonumber&\geq\ell^{-q}-2C_{d,\varepsilon_0}\ell^{-\left(\theta-\frac{d-1}{2}\right)}\geq\tfrac{1}{2}\ell^{-q},
\end{align}
$\Lambda_\ell(a)$ is polynomially level spacing for $H_{\Lambda_\ell(a)}$, and $q<\theta-\tfrac{d-1}{2}$. \eqref{diffun} follows from \cite[Lemma~3.3]{EK}.

If $\lambda\in\sigma_{\{a\}}(H_\Theta)$, we have $\lambda=\widetilde{\lambda}_x^{(a)}$ for some $x\in\Lambda_\ell^{\Theta,\ell'}(a)$, thus \eqref{psidec0} follows from \eqref{diffun} as $\varphi_x^{(a)}(y)=0$
for all $y\in\Theta\setminus\Lambda_\ell(a)$.

If $\lambda\in\sigma(H_\Theta)\setminus\sigma_{\{a\}}(H_\Theta)$, for all $x\in\Lambda_\ell^{\Theta,\ell'}(a)$ we have
\beq\label{distlambda2}
\left|\lambda-\lambda_x^{(a)}\right|\geq\left|\lambda-\widetilde{\lambda}_x^{(a)}\right|-\left|\widetilde{\lambda}_x^{(a)}-\lambda_x^{(a)}\right|\geq L^{-q}
-C_{d,\varepsilon_0}\ell^{-\left(\theta-\frac{d-1}{2}\right)}\geq\tfrac{1}{2}L^{-q},
\eeq
since $\Theta$ is $L$-polynomially level spacing for $H_\Theta$, we have \eqref{tildedist}, and $q<\gamma_1q<\theta-\tfrac{d-1}{2}$. Therefore \eqref{psidec1} follows from Lemma~\ref{psidecayall}(i).
(Note that \eqref{pdec01} follows from Lemma~\ref{psidecayall}(ii).)

Now let $\{\Lambda_\ell(a)\}_{a\in\mathcal{G}}$, where $\mathcal{G}\subset\R^d$ such that $\Lambda_\ell(a)\subset\Theta$ for all $a\in\mathcal{G}$, be a collection of $\theta$-polynomially localizing
boxes with corresponding eigensystems $\left\{(\varphi_x^{(a)},\lambda_x^{(a)})\right\}_{x\in\Lambda_\ell(a)}$. Let $\lambda\in\sigma(H_\Theta)$, $a,b\in\mathcal{G}$, $a\neq b$, $x\in\Lambda_\ell^{\Theta,\ell'}(a)$ and
$y\in\Lambda_\ell^{\Theta,\ell'}(b)$. Assume $\lambda_x^{(a)},\lambda_x^{(b)}\in\mathcal{E}_{\mathcal{G}}^\Theta(\lambda)$, then it follows from \eqref{diffun} that
\beq
\left\|\varphi_x^{(a)}-\varphi_y^{(b)}\right\|\leq4C_{d,\varepsilon_0}L^q\ell^{-\left(\theta-\frac{d-1}{2}\right)},
\eeq
thus
\beq\label{invarphi1}
\left|\left\langle\varphi_x^{(a)},\varphi_y^{(b)}\right\rangle\right|\geq\Re\left\langle\varphi_x^{(a)},\varphi_y^{(b)}\right\rangle\geq1-8C_{d,\varepsilon_0}^2L^{2q}\ell^{-2\left(\theta-\frac{d-1}{2}\right)}.
\eeq
On the other hand, \eqref{defploc} gives
\beq\label{invarphi2}
\|x-y\|\geq2\ell'\Longrightarrow\left|\left\langle\varphi_x^{(a)},\varphi_y^{(b)}\right\rangle\right|\leq(\ell+1)^d\ell^{-\theta}.
\eeq
Combining \eqref{invarphi1} and \eqref{invarphi2}, we conclude that
\beq
\lambda_x^{(a)},\lambda_x^{(b)}\in\mathcal{E}_{\mathcal{G}}^\Theta(\lambda)\Longrightarrow\|x-y\|<2\ell'.
\eeq

To prove \eqref{sigmaab}, let $a,b\in\mathcal{G}$, $a\neq b$. Assume $\Lambda_\ell(a)\cap\Lambda_\ell(b)=\emptyset$, then
\beq
x\in\Lambda_\ell^{\Theta,\ell'}(a)\qtx{and}y\in\Lambda_\ell^{\Theta,\ell'}(b)\Longrightarrow\|x-y\|\geq2\ell',
\eeq
thus it follows from \eqref{difxy} that $\sigma_{\{a\}}(H_\Theta)\cap\sigma_{\{b\}}(H_\Theta)=\emptyset$.

Parts (ii)(b) and (ii)(c) follow immediately from parts (i)(b) and (i)(c) respectively. To prove part (ii)(d), we let $P_{\mathcal{G}}$ be the orthogonal projection onto the span of
$\{\psi_\lambda;\lambda\in\sigma_{\mathcal{G}}(H_\Theta)\}$. \eqref{psidec4} gives
\beq
\|(1-P_{\mathcal{G}})\delta_y\|\leq C_{d,\varepsilon_0}L^q\ell^{-(\theta-2d)}|\Theta|^{\frac{1}{2}}\qtx{for all}y\in\Theta'_{\mathcal{G}},
\eeq
thus
\beq
\|(1-P_{\mathcal{G}})\chi_{\Theta'_{\mathcal{G}}}\|\leq|\Theta'_{\mathcal{G}}|^{\frac{1}{2}}|\Theta|^{\frac{1}{2}}C_{d,\varepsilon_0}L^q\ell^{-(\theta-2d)}
\leq|\Theta|C_{d,\varepsilon_0}L^q\ell^{-(\theta-2d)}.
\eeq
If $|\Theta|\leq(L+1)^d$, we have
\beq
\|(1-P_{\mathcal{G}})\chi_{\Theta'_{\mathcal{G}}}\|\leq(L+1)^dC_{d,\varepsilon_0}L^q\ell^{-(\theta-2d)}<1
\eeq
since $d+q<\gamma_1(d+q)<\theta-2d$, so it follows from \cite[Lemma~A.1]{EK} that
\beq
|\Theta'_{\mathcal{G}}|=\tr\chi_{\Theta'_{\mathcal{G}}}\leq\tr P_{\mathcal{G}}=|\sigma_{\mathcal{G}}(H_\Theta)|.
\eeq
Using a similar argument and \eqref{psidec3}, we can prove $|\sigma_{\mathcal{G}}(H_\Theta)|\leq|\Theta_{\mathcal{G}}|$.
\end{proof}

\subsection{Buffered subsets}\label{sectbuff}
For boxes $\Lambda_\ell\subset\Lambda_L$ that are not $\sharp$ for $H$, we will surround them with a buffer of $\sharp$ boxes and study eigensystems for the augmented subset.

\begin{definition}\label{defbufall}
Let $\Lambda_L=\Lambda_L(x_0)$ and $x_0\in\R^d$. $\Upsilon\subset\Lambda_L$ is called a $\sharp$-buffered subset of $\Lambda_L$, where $\sharp$ stands for $\theta$-PL,
$s$-SEL, $m^\ast$-ML or $m$-LOC, if the following holds:
\begin{enumerate}
\item $\Upsilon$ is a connected set in $\Z^d$ of the form
\beq
\Upsilon=\bigcup_{j=1}^J\Lambda_{R_j}(a_j)\cap\Lambda_L,
\eeq
where $J\in\N$, $a_1,a_2,\ldots,a_J\in\Lambda^\R_L$, and $\ell\leq R_j\leq L$ for $j=1,2,\ldots,J$.
\item $\Upsilon$ is $L$-polynomially level spacing for $H$ if $\sharp$ is $\theta$-PL or $m^\ast$-ML, $L$-level spacing for $H$ if $\sharp$ is $s$-SEL or $m$-LOC.
\item[(iii)]  There exists $\mathcal{G}_\Upsilon\subset\Lambda^\R_L$ such that:
\begin{enumerate}
\item For all $a\in\mathcal{G}_\Upsilon$ we have $\Lambda_\ell(a)\subset\Upsilon$, $\Lambda_\ell(a)$ is a $\sharp$ box for $H$.
\item For all $y\in\partial_\mathrm{in}^{\Lambda_L}\Upsilon$ there exists $a_y\in\mathcal{G}_\Upsilon$ such that $y\in\Lambda_\ell^{\Upsilon,2\ell_\sharp}(a_y)$.
\end{enumerate}
\end{enumerate}
In this case we set
\beq\label{defUpsck}
\widecheck{\Upsilon}=\bigcup_{a\in\mathcal{G}_\Upsilon}\Lambda_\ell(a),\quad\widecheck{\Upsilon'}=\bigcup_{a\in\mathcal{G}_\Upsilon}\Lambda_\ell^{\Upsilon, 2\ell_\sharp}(a),
\quad\widehat{\Upsilon}=\Upsilon\setminus\widecheck{\Upsilon},\qtx{and}\widehat{\Upsilon'}=\Upsilon\setminus\widecheck{\Upsilon'}.
\eeq
($\widecheck{\Upsilon}=\Upsilon_{\mathcal{G}_\Upsilon}$ and $\widecheck{\Upsilon'}=\Upsilon'_{\mathcal{G}_\Upsilon}$ in the notation of Lemma~\ref{ideigsysall}.)
\end{definition}

\begin{lemma}\label{lembuffall}
Given a $\sharp$-buffered subset $\Upsilon$ of $\Lambda_L$, let $\{(\psi_\nu,\nu)\}_{\nu\in\sigma(H_\Upsilon)}$ be an eigensystem for $H_\Upsilon$. Let
$\mathcal{G}=\mathcal{G}_\Upsilon$ and set
\beq\label{sigmabad}
\sigma_{\mathcal{B}}(H_{{\Upsilon}})=\sigma(H_\Upsilon)\setminus\sigma_{\mathcal{G}}(H_\Upsilon),
\eeq
where $\sigma_{\mathcal{G}}(H_\Upsilon)$ is as in \eqref{defE}. Then the following holds for sufficiently large $\ell$:
\begin{enumerate}
\item If $\nu\in\sigma_{\mathcal{B}}(H_\Upsilon)$ we have for all $y\in\widecheck{\Upsilon'}$:
\beq\label{psidec5}
    |\psi_\lambda(y)|\leq\begin{cases}
    C_{d,\varepsilon_0}L^q\ell^{-(\theta-2d)}&\text{if}\;\sharp\sqtx{is}\theta\text{-PL}\\
    \e^{-c_2\ell^s},\sqtx{with}c_2=c_2(\ell)\sqtx{as in \eqref{sc2}}&\text{if}\;\sharp\sqtx{is}s\text{-SEL}\\
    \e^{-m^\ast_2\ell_\tau},\sqtx{with}m^\ast_2=m^\ast_2(\ell)\sqtx{as in \eqref{mst2}}&\text{if}\;\sharp\sqtx{is}m^\ast\text{-ML}\\
    \e^{-m_2\ell_\tau},\sqtx{with}m_2=m_2(\ell)\sqtx{as in \eqref{m2}}&\text{if}\;\sharp\sqtx{is}m\text{-LOC}
    \end{cases},
    \eeq
and
\beq
\left|\widehat{\Upsilon}\right|\leq|\sigma_{\mathcal{B}}(H_\Upsilon)|\leq\left|\widehat{\Upsilon'}\right|.
\eeq
\item Let $\Lambda_L$ be polynomially level spacing for $H$ if $\sharp$ is $\theta$-PL or $m^\ast$-ML, level spacing for $H$ if $\sharp$ is $s$-SEL or $m$-LOC, and let
$\{(\phi_\lambda,\lambda)\}_{\lambda\in\sigma(H_{\Lambda_L})}$ be an eigensystem for $H_{\Lambda_L}$. There exists an injection
\beq\label{injectbad}
\nu\in\sigma_{\mathcal{B}}(H_\Upsilon)\mapsto\widetilde{\nu}\in\sigma(H_{\Lambda_L})\setminus\sigma_{\mathcal{G}}(H_{\Lambda_L}),
\eeq
such that for all $\nu\in\sigma_{\mathcal{B}}(H_\Upsilon)$:
 \begin{enumerate}
    \item If $\sharp$ is $\theta$-PL, we have
    \beq\label{distnu}
    |\widetilde{\nu}-\nu|\leq C_{d,\varepsilon_0}L^{\frac{d}{2}+q}\ell^{-(\theta-2d)},
    \eeq
    and, multiplying each $\psi_\nu$ by a suitable phase factor,
    \beq\label{diffun2}
    \|\phi_{\widetilde{\nu}}-\psi_\nu\|\leq2C_{d,\varepsilon_0}L^{\frac{d}{2}+2q}\ell^{-(\theta-2d)}.
    \eeq
    \item If $\sharp$ is $s$-SEL, we have
    \beq\label{distnusub}
    |\widetilde{\nu}-\nu|\leq\e^{-c_3\ell^s},\sqtx{where}c_3=c_3(\ell)\geq1-C_{d,\varepsilon_0}L^\beta\ell^{-s},
    \eeq
    and, multiplying each $\psi_\nu$ by a suitable phase factor,
    \beq\label{diffun2sub}
    \|\phi_{\widetilde{\nu}}-\psi_\nu\|\leq2\e^{-c_3\ell^s}\e^{L^\beta}.
    \eeq
    \item If $\sharp$ is $m^\ast$-ML, we have
    \beq\label{distnustr}
    |\widetilde{\nu}-\nu|\leq\e^{-m^\ast_4\ell_\tau},\sqtx{where}m^\ast_4=m^\ast_4(\ell)\geq m^\ast-C_{d,\varepsilon_0}\gamma_1q\tfrac{\log\ell}{\ell_\tau},
    \eeq
    and, multiplying each $\psi_\nu$ by a suitable phase factor,
    \beq\label{diffun2str}
    \|\phi_{\widetilde{\nu}}-\psi_\nu\|\leq2\e^{-m^\ast_4\ell_\tau}L^q.
    \eeq
    \item If $\sharp$ is $m$-LOC, we have
    \beq
    |\widetilde{\nu}-\nu|\leq\e^{-m_4\ell_\tau},\sqtx{where}m_4=m_4(\ell)\geq m-C_{d,\varepsilon_0}\ell^{\gamma\beta-\tau},
    \eeq
    and, multiplying each $\psi_\nu$ by a suitable phase factor,
    \beq\label{diffun2loc}
    \|\phi_{\widetilde{\nu}}-\psi_\nu\|\leq2\e^{-m_4\ell_\tau}\e^{L^\beta}.
    \eeq
    \end{enumerate}
\end{enumerate}
\end{lemma}

\begin{proof}
Part (i) follows immediately from Lemma~\ref{ideigsysall}(ii)(c) and (ii)(d).

Let $\Lambda_L$ be polynomially level spacing, and let $\{(\phi_\lambda,\lambda)\}_{\lambda\in\sigma(H_{\Lambda_L})}$ be an eigensystem for $H_{\Lambda_L}$.
It follows from \cite[Lemma~3.2]{EK} that for $\nu\in\sigma_{\mathcal{B}}(H_\Upsilon)$ we have
\begin{align}
\|(H_{\Lambda_L}-\nu)\psi_\nu\|&\leq(2d-1)\varepsilon|\partial_{\mathrm{ex}}^{\Lambda_L}\Upsilon|^{\frac{1}{2}}\left\|\varphi_{\partial_{\mathrm{in}}^{\Lambda_L}\Upsilon}\right\|_\infty
\leq(2d-1)\varepsilon L^{\frac{d}{2}}C_{d,\varepsilon_0}L^q\ell^{-(\theta-2d)}
\\\nonumber&\leq C_{d,\varepsilon_0}L^{\frac{d}{2}+q}\ell^{-(\theta-2d)},
\end{align}
where we used
$\partial_{\mathrm{in}}^{\Lambda_L}\Upsilon\subset\widecheck{\Upsilon'}$ and \eqref{psidec5}. The map in \eqref{injectbad} is a well defined injection into $\sigma(H_{\Lambda_L})$
since $\Lambda_L$ and $\Upsilon$ are $L$-polynomially level spacing for $H$, and \eqref{diffun2} follows from \eqref{distnu} and \cite[Lemma~3.3]{EK}.

To show $\widetilde\nu\not\in \sigma_{\mathcal{G}}(H_{\Lambda_L})$ for all $\nu\in\sigma_{\mathcal{B}}(H_\Upsilon)$, we assume $\widetilde{\nu_1}\in\sigma_{\mathcal{G}}(H_{\Lambda_L})$
for some $\nu_1\in\sigma_{\mathcal{B}}(H_\Upsilon)$. Then there is $a\in\mathcal{G}$ and $x\in\Lambda_\ell^{\Lambda_L,\ell'}(a)$ such that
$\lambda_x^{(a)}\in\mathcal{E}_{\mathcal{G}}^{\Lambda_L}(\widetilde{\nu_1})$.  On the other hand, $\lambda_x^{(a)}\in\mathcal{E}_{\mathcal{G}}^{\Upsilon}(\lambda_1)$ for some
$\lambda_1\in\sigma_{\mathcal{G}}(H_\Upsilon)$ by Lemma~\ref{ideigsysall}(i)(a). We conclude from \eqref{diffun} and \eqref{diffun2} that
\begin{align}
\sqrt{2}=\|\psi_{\lambda_1}-\psi_{\nu_1}\|&\leq\left\|\psi_{\lambda_1}-\varphi_x^{(a)}\right\|+\left\|\varphi_x^{(a)}-\phi_{\widetilde{\nu_1}}\right\|
+\left\|\phi_{\widetilde{\nu_1}}-\psi_{\nu_1}\right\|
\\\nonumber&\leq4C_{d,\varepsilon_0}L^q\ell^{-\left(\theta-\frac{d-1}{2}\right)}+2C_{d,\varepsilon_0}L^{\frac{d}{2}+2q}\ell^{-(\theta-2d)}<1,
\end{align}
a contradiction.
\end{proof}

\begin{lemma}\label{lembadall}
Given $\Lambda_L=\Lambda_L(x_0)$, $x_0\in\R^d$, let $\Upsilon$ be a $\sharp$-buffered subset of $\Lambda_L$. Let $\mathcal{G}=\mathcal{G}_\Upsilon$ and set
\begin{align}\label{defEUps}
\mathcal{E}_{\mathcal{G}}^{\Lambda_L}(\nu)&=\left\{\lambda_x^{(a)};a\in\mathcal{G},x\in\Lambda_\ell^{\Lambda_L,\ell_\sharp}(a),\widetilde{\lambda}_x^{(a)}=\nu\right\}
\subset\mathcal{E}_{\mathcal{G}}^\Upsilon(\nu)\sqtx{for}\nu\in\sigma(H_\Upsilon),
\\\nonumber\sigma_{\mathcal{G}}^{\Lambda_L}(H_\Upsilon)&=\left\{\nu\in\sigma(H_\Upsilon);\mathcal{E}_{\mathcal{G}}^{\Lambda_L}(\lambda)\neq\emptyset\right\}
\subset\sigma_{\mathcal{G}}(H_\Upsilon).
\end{align}
The following holds for sufficiently large $\ell$:
\begin{enumerate}
\item Let $(\psi,\lambda)$ be an eigenpair for $H_{\Lambda_L}$ such that for all $\nu\in\sigma_{\mathcal{G}}^{\Lambda_L}(H_\Upsilon)\cup\sigma_{\mathcal{B}}(H_{\Upsilon})$,
\beq\label{disteig2}
|\lambda-\nu|\geq\begin{cases}
\frac{1}{2}L^{-q}&\text{if}\;\sharp\sqtx{is}\theta\text{-PL}\sqtx{or}m^\ast\text{-ML}\\
\frac{1}{2}\e^{-L^\beta}&\text{if}\;\sharp\sqtx{is}s\text{-SEL}\sqtx{or}m\text{-LOC}
\end{cases}.
\eeq
Then for all $y\in\Upsilon^{\Lambda_L,2\ell_\sharp}$:
\begin{enumerate}
\item If $\sharp$ is $\theta$-PL, we have
\beq\label{psidec6}
|\psi(y)|\leq C_{d,\varepsilon_0}L^{2d+2q}\ell^{-(\theta-2d)}|\psi(v)|\qtx{for some}v\in\partial^{\Lambda_L,2\ell'}\Upsilon.
\eeq
\item If $\sharp$ is $s$-SEL, we have
\begin{gather}\label{psidec6sub}
|\psi(y)|\leq\e^{-c_4\ell^s}|\psi(v)|\qtx{for some}v\in\partial^{\Lambda_L,2\ell'}\Upsilon,
\\\label{sc4}\text{where}\quad
c_4=c_4(\ell)\geq1-C_{d,\varepsilon_0}L^\beta\ell^{-s}.
\end{gather}
\item If $\sharp$ is $m^\ast$-ML, we have
\begin{gather}\label{psidec6str}
|\psi(y)|\leq\e^{-m^\ast_5\ell_\tau}|\psi(v)|\qtx{for some}v\in\partial^{\Lambda_L,2\ell_\tau}\Upsilon,
\\\label{mst5}\text{where}\quad
m^\ast_5=m^\ast_5(\ell)\geq m^\ast-C_{d,\varepsilon_0}\gamma_1q\tfrac{\log\ell}{\ell_\tau}.
\end{gather}
\item If $\sharp$ is $m$-LOC, we have
\begin{gather}\label{psidec6loc}
|\psi(y)|\leq\e^{-m_5\ell_\tau}|\psi(v)|\qtx{for some}v\in\partial^{\Lambda_L,2\ell_\tau}\Upsilon,
\\\label{m5}\text{where}\quad
m_5=m_5(\ell)\geq m-C_{d,\varepsilon_0}\ell^{\gamma\beta-\tau}.
\end{gather}
\end{enumerate}

\item Let $\Lambda_L$ be polynomially level spacing for $H$ if $\sharp$ is $\theta$-PL or $m^\ast$-ML, level spacing for $H$ if $\sharp$ is $s$-SEL or $m$-LOC. Let
$\{(\psi_\lambda,\lambda)\}_{\lambda\in\sigma(H_{\Lambda_L})}$ be an eigensystem for $H_{\Lambda_L}$, and set (recalling \eqref{injectbad})
\beq
\sigma_\Upsilon(H_{\Lambda_L})=\{\widetilde\nu;\nu\in\sigma_{\mathcal{B}}(H_\Upsilon)\}\subset\sigma(H_{\Lambda_L})\setminus\sigma_{\mathcal{G}}(H_{\Lambda_L}).
\eeq
Then the condition \eqref{disteig2} is satisfied for all $\lambda\in\sigma(H_{\Lambda_L})\setminus(\sigma_{\mathcal{G}}(H_{\Lambda_L})\cup\sigma_\Upsilon(H_{\Lambda_L}))$,
so for all $y\in\Upsilon^{\Lambda_L,2\ell_\sharp}$
\beq
    |\psi_\lambda(y)|\leq\begin{cases}
    C_{d,\varepsilon_0}L^{2d+2q}\ell^{-(\theta-2d)}|\psi(v)|&\text{if}\;\sharp\sqtx{is}\theta\text{-PL}\\
   \e^{-c_4\ell^s}|\psi(v)|&\text{if}\;\sharp\sqtx{is}s\text{-SEL}\\
    \e^{-m^\ast_5\ell_\tau}|\psi(v)|&\text{if}\;\sharp\sqtx{is}m^\ast\text{-ML}\\
   \e^{-m_5\ell_\tau}|\psi(v)|&\text{if}\;\sharp\sqtx{is}m\text{-LOC}
    \end{cases}
    \eeq
\end{enumerate}
for some $v\in\partial^{\Lambda_L,2\ell_\sharp}\Upsilon$.
\end{lemma}

\begin{proof}
Let $\{(\vartheta_\nu,\nu)\}_{\nu\in\sigma(H_\Upsilon)}$ be an eigensystem for $H_\Upsilon$. For $\nu\in\sigma_{\mathcal{G}}(H_\Upsilon)$ we fix
$\lambda_{x_\nu}^{(a_\nu)}\in\mathcal{E}_{\mathcal{G}}^\Upsilon(\nu)$, where $a_\nu\in\mathcal{G}$, $x_\nu\in\Lambda_\ell^{\Upsilon,\ell'}(a_\nu)$.
If $\nu\in\sigma_{\mathcal{G}}^{\Lambda_L}(H_\Upsilon)$, we choose $\lambda_{x_\nu}^{(a_\nu)}\in\mathcal{E}_{\mathcal{G}}^{\Lambda_L}(\nu)$, thus $x_\nu\in \Lambda_\ell^{\Lambda_L,\ell'}(a_\nu)$.
If $\nu\in\sigma_{\mathcal{G}}(H_\Upsilon)\setminus\sigma_{\mathcal{G}}^{\Lambda_L}(H_\Upsilon)$ we have
$ x_\nu\in\Lambda_\ell^{\Upsilon,\ell'}(a_\nu)\setminus\Lambda_\ell^{\Lambda_L,\ell'}(a_\nu)$.

Given $y\in\Upsilon$, we have (see \eqref{eigsys})
\begin{align}\label{sum6}
\psi(y)&=\sum_{\nu\in\sigma(\Upsilon)}\vartheta_\nu(y)\langle\vartheta_\nu,\psi\rangle
\\\nonumber&=\sum_{\nu\in\sigma_{\mathcal{G}}^{\Lambda_L}(H_\Upsilon)\cup\sigma_{\mathcal{B}}(H_\Upsilon)}\vartheta_\nu(y)\langle\vartheta_\nu,\psi\rangle
+\sum_{\nu\in\sigma_{\mathcal{G}}(H_\Upsilon)\setminus\sigma_{\mathcal{G}}^{\Lambda_L}(H_\Upsilon)}\vartheta_\nu(y)\langle\vartheta_\nu,\psi\rangle.
\end{align}

Let $(\psi,\lambda)$ be an eigenpair for $H_{\Lambda_L}$ satisfying \eqref{disteig2}. If $\nu\in\sigma_{\mathcal{G}}^{\Lambda_L}(H_\Upsilon)\cup\sigma_{\mathcal{B}}(H_\Upsilon)$, we have
\beq
\langle\vartheta_\nu,\psi\rangle=(\lambda-\nu)^{-1}\langle\vartheta_\nu,(H_{\Lambda_L}-\nu)\psi\rangle=(\lambda-\nu)^{-1}\langle(H_{\Lambda_L}-\nu)\vartheta_u,\psi\rangle.
\eeq
It follows from \eqref{disteig2} and \cite[Lemma~3.2]{EK} that
\begin{align}
|\vartheta_\nu(y)\langle\vartheta_\nu,\psi\rangle|&\leq2L^q\varepsilon|\vartheta_\nu(y)|\sum_{v\in\partial_{\mathrm{ex}}^{\Lambda_L}\Upsilon}
\left(\sum_{v'\in\partial_{\mathrm{in}}^{\Lambda_L}\Upsilon,|v'-v|=1}|\vartheta_\nu(v')|\right)|\psi(v)|
\\\nonumber&\leq2\varepsilon L^{q+d}\left(2d\max_{u\in\partial_{\mathrm{in}}^{\Lambda_L}\Upsilon}|\vartheta_\nu(u)|\right)|\psi(v_1)|\qtx{for some}v_1\in\partial_{\mathrm{ex}}^{\Lambda_L}\Upsilon.
\end{align}
If $\nu\in\sigma_{\mathcal{B}}(H_\Upsilon)$, \eqref{psidec5} gives
\beq
\max_{u\in\partial_{\mathrm{in}}^{\Lambda_L}\Upsilon}|\vartheta_\nu(u)|\leq C_{d,\varepsilon_0}L^q\ell^{-(\theta-2d)}.
\eeq
If $\nu\in\sigma_{\mathcal{G}}^{\Lambda_L}(H_\Upsilon)$, it follows from \eqref{diffun} and \eqref{defploc}, that
\begin{align}
&\max_{u\in\partial_{\mathrm{in}}^{\Lambda_L}\Upsilon}|\vartheta_\nu(u)|
\leq\max_{u\in\partial_{\mathrm{in}}^{\Lambda_L}}\left(\left|\vartheta_\nu(u)-\varphi_{x_\nu}^{(a_\nu)}\right|+\left|\varphi_{x_\nu}^{(a_\nu)}\right|\right)
\\\nonumber&\quad\leq2C_{d,\varepsilon_0}L^q\ell^{-\left(\theta-\frac{d-1}{2}\right)}+\ell^{-\theta}\leq 3C_{d,\varepsilon_0}L^q\ell^{-\left(\theta-\frac{d-1}{2}\right)}\leq C_{d,\varepsilon_0}L^q\ell^{-(\theta-2d)}.
\end{align}
Therefore (recalling \eqref{ca1}),
\begin{align}\label{sum7}
\left|\sum_{\nu\in\sigma_{\mathcal{G}}^{\Lambda_L}(H_\Upsilon)\cup\sigma_{\mathcal{B}}(H_\Upsilon)}\vartheta_\nu(y)\langle\vartheta_\nu,\psi\rangle\right|
&\leq4d\varepsilon L^{2d+q}\left(C_{d,\varepsilon_0}L^q\ell^{-(\theta-2d)}\right)|\psi(v_2)|
\\\nonumber&\leq C_{d,\varepsilon_0}L^{2d+2q}\ell^{-(\theta-2d)}|\psi(v_2)|,
\end{align}
for some $v_2\in\partial_{\mathrm{ex}}^{\Lambda_L}\Upsilon$.

If $\nu\in\sigma_{\mathcal{G}}(H_\Upsilon)\setminus\sigma_{\mathcal{G}}^{\Lambda_L}(H_\Upsilon)$, we have $x_\nu\in \Lambda_\ell^{\Upsilon,\ell'}(a_\nu)\setminus\Lambda_\ell^{\Lambda_L,\ell'}(a_\nu)$,
thus
\beq
\dist(x_\nu,\Upsilon\setminus\Lambda_\ell(a_\nu))>\ell'\qtx{and}\dist(x_\nu,\Lambda_L\setminus\Lambda_\ell(a_\nu))\leq\ell',
\eeq
 and hence there is $u_0\in\Lambda_L\setminus\Upsilon$ such that $\|x_\nu-u_0\|\leq\ell'$. We suppose $y\in\Upsilon^{\Lambda_L,2\ell'}$, then $\|y-u_0\|>2\ell'$. Therefore
\beq
\|x_\nu-y\|\geq\|y-u_0\|-\|x_\nu-u_0\|>2\ell'-\ell'=\ell'.
\eeq
Thus it follows from \eqref{diffun} and \eqref{defploc} that
\begin{align}
|\vartheta_\nu(u)|&\leq\left|\vartheta_\nu(u)-\varphi_{x_\nu}^{(a_\nu)}\right|+\left|\varphi_{x_\nu}^{(a_\nu)}\right|\leq 2C_{d,\varepsilon_0}L^q\ell^{-\left(\theta-\frac{d-1}{2}\right)}+\ell^{-\theta}
\\\nonumber&\leq3C_{d,\varepsilon_0}L^q\ell^{-\left(\theta-\frac{d-1}{2}\right)}.
\end{align}
Therefore
\beq\label{sum8}
\left|\sum_{\nu\in\sigma_{\mathcal{G}}(H_\Upsilon)\setminus\sigma_{\mathcal{G}}^{\Lambda_L}(H_\Upsilon)}\vartheta_\nu(y)\langle\vartheta_\nu,\psi\rangle\right|\leq
3C_{d,\varepsilon_0}L^q(L+1)^{\frac{3d}{2}}\ell^{-\left(\theta-\frac{d-1}{2}\right)}|\psi(v_3)|,
\eeq
for some $v_3\in\Upsilon$.

Combining \eqref{sum6}, \eqref{sum7} and \eqref{sum8}, we conclude that for all $y\in\Upsilon^{\Lambda_L,2\ell'}$,
\beq
|\psi(y)|\leq C_{d,\varepsilon_0}L^{2d+2q}\ell^{-(\theta-2d)}|\psi(v_4)|,
\eeq
for some $v_4\in\Upsilon\cup\partial_{\mathrm{ex}}^{\Lambda_L}\Upsilon$. If $v_4\in\Upsilon^{\Lambda_L,2\ell'}$ we repeat the procedure to estimate $|\psi(v_4)|$. Since we can suppose $\psi(y)\neq0$
without loss of generality, the procedure must stop after finitely many times, and at that time we must have \eqref{psidec6}.

Now let $\Lambda_L$ be polynomially level spacing. If $\lambda\not\in\sigma_{\mathcal{G}}(H_{\Lambda_L})$, it follows from Lemma \ref{ideigsysall}(i)(c) that \eqref{distlambda} holds for all $a\in\mathcal{G}$.
If $\lambda\not\in\sigma_\Upsilon(H_{\Lambda_L})$, using the argument in \eqref{distlambda2}, with \eqref{distnu} instead of \eqref{tildedist}, we get $|\lambda-\nu|\geq\tfrac{1}{2}L^{-q}$ for all
$\nu\in\sigma_{\mathcal{B}}(H_\Upsilon)$. Therefore we have \eqref{disteig2}, which implies \eqref{psidec6}.
\end{proof}

\section{Probability estimates}\label{secprob}
The following lemma gives the probability estimates for polynomially level spacing and level spacing.
\begin{lemma}\label{lemlsp}
Let $H_{\varepsilon,\omega}$ be the Anderson model. Let $\Theta\subset\Z^d$ and $L>1$. Then, for all $\varepsilon\leq\varepsilon_0$,
\beq
\P\{\Theta\sqtx{is}L\text{-polynomially level spacing for}\; H\}\geq1-Y_{\varepsilon_0}L^{-(2\alpha-1)q}|\Theta|^2,
\eeq
and
\beq
\P\{\Theta\sqtx{is}L\text{-level spacing for}\; H\}\geq1-Y_{\varepsilon_0}\e^{-(2\alpha-1)L^\beta}|\Theta|^2,
\eeq
where
\beq
Y_{\varepsilon_0}=2^{2\alpha-1}\widetilde{K}^2(\diam\supp\mu+2d\varepsilon_0+1),
\eeq
with $\widetilde{K}=K$ if $\alpha=1$ and $\widetilde{K}=8K$ if $\alpha\in\left(\tfrac{1}{2},1\right)$.
\end{lemma}

Lemma~\ref{lemlsp} follows from \cite[Lemma~2.1]{EK} and its proof. (Also see \cite[Lemma~2]{KlM}.)

\section{Bootstrap multiscale analysis}\label{sectbmsa}

In this section, we fix $\theta>\left(\tfrac{6}{2\alpha-1}+\tfrac{9}{2}\right)d$ and $0<\xi<1$. (Note that Proposition~\ref{propmsa1} is independent of $\xi$.)
We will omit the dependence on $\theta$ and $\xi$ from the notation. We denote the complementary event of an event $\mathcal{E}$ by $\mathcal{E}^c$.

\subsection{The first multiscale analysis}

\begin{proposition}\label{propmsa1}
Fix $\varepsilon_0>0$, $Y\geq400$, and $P_0<\tfrac{1}{2}(2Y)^{-2d}$. There exists a finite scale $\mathcal{L}(\varepsilon_0,Y)$ with the following property:
Suppose for some scale $L_0\geq\mathcal{L}(\varepsilon_0,Y)$, and $0<\varepsilon\leq\varepsilon_0$ we have
\beq\label{initcon}
\inf_{x\in\R^d}\P\{\Lambda_{L_0}(x)\sqtx{is}\theta\text{-polynomially localizing for}\;H_{\varepsilon,\omega}\}\geq1-P_0.
\eeq
Then, setting $L_{k+1}=YL_k$ for $k=0,1,\ldots$, there exists $K_0=K_0(Y,L_0,P_0)\in\N$ such that
\beq\label{res1}
\inf_{x\in\R^d}\P\{\Lambda_{L_k}(x)\sqtx{is}\theta\text{-polynomially localizing for}\;H_{\varepsilon,\omega}\}\geq1-L_k^{-p}\sqtx{for}k\geq K_0.
\eeq
\end{proposition}

Proposition~\ref{propmsa1} follows from the following induction step for the multiscale analysis.

\begin{lemma}\label{indumsa1}
Fix $\varepsilon_0>0$, $Y\geq400$, and $P\leq1$. Suppose for some scale $\ell$  and $0<\varepsilon\leq\varepsilon_0$ we have
\beq\label{hypmsaind}
\inf_{x\in\R^d}\P\{\Lambda_\ell(x)\sqtx{is}\theta\text{-polynomially localizing for}\;H_{\varepsilon,\omega}\}\geq1-P.
\eeq
Then, if $\ell$ is sufficiently large, for $L=Y\ell$ we have
\beq
\inf_{x\in\R^d}\P\{\Lambda_L(x)\sqtx{is}\theta\text{-polynomially localizing for}\;H_{\varepsilon,\omega}\}\geq1-\left((2Y)^{2d}P^2+\tfrac{1}{2}L^{-p}\right).
\eeq
\end{lemma}

\begin{proof}
We fix $0<\varepsilon\leq\varepsilon_0$ and suppose \eqref{hypmsaind} for some scale $\ell$. Let $\Lambda_L=\Lambda_L(x_0)$, where $x_0\in\R^d$, and let
$\mathcal{C}_{L,\ell}=\mathcal{C}_{L,\ell}(x_0)$ be the suitable $\ell$-cover of $\Lambda_L$. For $N\in\N$, let $\mathcal{B}_N$ denote the event that there exist
at most $N$ disjoint boxes in $\mathcal{C}_{L,\ell}$ that are not $\theta$-PL for $H_{\varepsilon,\omega}$. Using \eqref{hypmsaind}, \eqref{covernum} and the fact
that events on disjoint boxes are independent, if $N=1$ we have
\beq\label{probb}
\P\{\mathcal{B}_N^c\}\leq\left(\tfrac{2L}{\ell}\right)^{(N+1)d}P^{N+1}=(2Y)^{(N+1)d}P^{N+1}=(2Y)^{2d}P^2.
\eeq

We now fix $\omega\in\mathcal{B}_N$. There exists $\mathcal{A}_N=\mathcal{A}_N(\omega)\in\Xi_{L,\ell}=\Xi_{L,\ell}(x_0)$, with $|\mathcal{A}_N|\leq N$ and
$\|a-b\|\geq2\rho\ell$ (i.e., $\Lambda_\ell(a)\cap\Lambda_\ell(b)=\emptyset$) if $a,b\in\mathcal{A}_N$, $a\neq b$, such that for all $a\in\Xi_{L,\ell}$ with
$\dist(a,\mathcal{A}_N)\geq2\rho\ell$ (i.e., $\Lambda_\ell(a)\cap\Lambda_\ell(b)=\emptyset$ for all $b\in\mathcal{A}_N$), $\Lambda_\ell(a)$ is a $\sharp$ box for $H_{\varepsilon,\omega}$
($\sharp$ stands for $\theta$-PL). In other words,
\beq\label{imploc}
a\in\Xi_{L,\ell}\setminus\bigcup_{b\in\mathcal{A}_N}\Lambda^{\R}_{(2\rho+1)\ell}(a_0)\Longrightarrow\Lambda_\ell(a)\sqtx{is a}\sharp\sqtx{box for}H_{\varepsilon,\omega}.
\eeq

To embed the box $\{\Lambda_\ell(b)\}_{b\in\mathcal{A}_N}$ into $\sharp$-buffered subsets of $\Lambda_L$, we consider graphs  $\mathbb{G}_i=(\Xi_{L,\ell},\mathbb{E}_i)$, $i=1,2$,
both having $\Xi_{L,\ell}$ as the set of vertices, with sets of edges given by
\begin{align}
\mathbb{E}_1&=\{\{a,b\}\in\Xi_{L,\ell}^2;\|a-b\|=\rho\ell\}
\\\nonumber&=\{\{a,b\}\in\Xi_{L,\ell}^2;a\neq b\sqtx{and}\Lambda_\ell(a)\cap\Lambda_\ell(b)\neq\emptyset\},
\\\nonumber\mathbb{E}_2&=\{\{a,b\}\in \Xi_{L,\ell}^2;\sqtx{either}\|a-b\|=2\rho\ell\sqtx{or}\|a-b\|=3\rho\ell\}
\\\nonumber&=\{\{a,b\}\in\Xi_{L,\ell}^2;\Lambda_\ell(a)\cap\Lambda_\ell(b)=\emptyset\sqtx{and}\Lambda_{(2\rho+1)\ell}(a)\cap\Lambda_{(2\rho+1)\ell}(b)\neq\emptyset\}.
\end{align}
Let $\{\Phi_r\}_{r=1}^R=\{\Phi_r(\omega)\}_{r=1}^R$ denote the $\mathbb{G}_2$-connected components of $\mathcal{A}_{N}$ (i.e., connected in the graph $\mathcal{G}_2$). Note that
\beq
R\in\{1,2,\ldots,N\},\quad\sum_{r=1}^R|\Phi_r|=|\mathcal{A}_N|\leq N,\qtx{and}\diam\Phi_r\leq3\rho\ell(|\Phi_r|-1).
\eeq
Set
\beq
\widetilde{\Phi}_r=\Xi_{L,\ell}\cap\bigcup_{a\in\Phi_r}\Lambda^\R_{(2\rho+1)\ell}(a)=\{a\in\Xi_{L,\ell};\dist(a,\Phi_r)\leq\rho\ell\},
\eeq
and note that $\left\{\widetilde{\Phi}_r\right\}_{r=1}^R$ is a collection of disjoint, $\mathbb{G}_1$-connected subsets of $ \Xi_{L,\ell}$, such that
\beq
\diam\widetilde{\Phi}_r\leq\diam{\Phi}_r+2\rho\ell\leq\rho\ell(3|\Phi_r|-1)\sqtx{and}\dist(\widetilde{\Phi}_r,\widetilde{\Phi}_{\widetilde{r}})\geq2\rho\ell,\;r\neq \widetilde{r}.
\eeq
Moreover, \eqref{imploc} gives
\beq\label{imploc2}
a\in\mathcal{G}=\mathcal{G}(\omega)=\Xi_{L,\ell}\setminus\bigcup_{r=1}^R\widetilde{\Phi}_r\quad\Longrightarrow\quad\Lambda_\ell(a)\sqtx{is a}\sharp\sqtx{box for}H_{\varepsilon,\omega}.
\eeq

For $\Psi\subset\Xi_{L,\ell}$, we define the exterior boundary of $\Psi$ in the graph $\mathbb{G}_1$ by
\beq
\partial_{\mathrm{ex}}^{\mathbb{G}_1}\Psi=\{a\in\Xi_{L,\ell};\dist(a,\Psi)=\rho\ell\}.
\eeq
It follows from \eqref{imploc2} that $\Lambda_\ell(a)$ is $\sharp$ for $H_{\varepsilon,\omega}$ for all $a\in\partial_{\mathrm{ex}}^{ \mathbb{G}_1}\widetilde{\Phi}_r$, $r=1,2,\ldots,R$.
Set $\overline{\Psi}= \Psi\cup\partial_{\mathrm{ex}}^{\mathbb{G}_1}\Psi$, and
set, for $r=1,2,\ldots,R$,
\begin{align}\label{conups}
\Upsilon_r^{(0)}&=\Upsilon_r^{(0)}(\omega)=\bigcup_{a\in{\widetilde{\Phi}}_r}\Lambda_\ell(a),
\\\nonumber\Upsilon_r&=\Upsilon_r(\omega)=\Upsilon_r^{(0)}\cup\bigcup_{a\in\partial_{\mathrm{ex}}^{\mathbb{G}_1}\widetilde{\Phi}_r}\Lambda_\ell(a)=\bigcup_{a\in\overline{\widetilde{\Phi}_r}}\Lambda_\ell(a).
\end{align}
Each $\Upsilon_r$, $r=1,2,\ldots,R$, satisfies all the requirements to be a $\theta$-PL-buffered subset of $\Lambda_L$ with $\mathcal{G}_{\Upsilon_r}=\partial_{\mathrm{ex}}^{\mathbb{G}_1}\widetilde{\Phi}_r$
(see Definition~\ref{defbufall}), except that we do not know if $\Upsilon_r$ is $L$-polynomially level spacing for $H_{\varepsilon,\omega}$. (Note that the sets $\{\Upsilon_r^{(0)}\}_{r=1}^R$ are disjoint,
but the sets $\{\Upsilon_r\}_{r=1}^R$ are not necessarily disjoint.) Note also that
\beq
\diam\overline{\widetilde{\Phi}}_r\leq\diam\widetilde{\Phi}_r+2\rho\ell\leq\rho\ell(3|\Phi_r|+1),
\eeq
and hence
\beq\label{diamups}
\diam\Upsilon_r\leq\diam\overline{\widetilde{\Phi}}_r+\ell\leq\rho\ell(3|\Phi_r|+1)+\ell\leq5\ell|\Phi_r|,
\eeq
thus
\beq\label{sumdiam}
\sum_{r=1}^R\diam\Upsilon_r\leq5\ell N.
\eeq

We can arrange for $\{\Upsilon_r\}_{r=1}^R$ to be a collection of $\theta$-PL-buffered subsets of $\Lambda_L$ as follows. It follows from Lemma~\ref{lemlsp} that for any $\Theta\subset\Lambda_L$ we have
\begin{align}\label{probls}
\P\{\Theta\sqtx{is}L\text{-polynomially level spacing for}\;H_{\varepsilon,\omega}\}\geq1-Y_{\varepsilon_0}\e^{-(2\alpha-1)L^\beta}(L+1)^{2d}.
\end{align}
Given a $\mathbb{G}_2$-connected subset $\Phi$ of $\Xi_{L,\ell}$, let $\Upsilon(\Phi)\subset\Lambda_L$ be constructed from $\Phi$ as in \eqref{conups}. Set
\begin{align}
\mathcal{F}_N=\bigcup_{r=1}^N\mathcal{F}(r),\sqtx{where}\mathcal{F}(r)=\{\Phi\subset\Xi_{L,\ell};\Phi\sqtx{is}\mathbb{G}_2\text{-connected}\sqtx{and} |\Phi|=r\}.
\end{align}
Let $\mathcal{F}(r,a)=\{\Phi\in\mathcal{F}_r;a\in\Phi\}$ for $a\in\Xi_{L,\ell}$, and note that each vertex in the graph $\mathbb{G}_2$ has less than $d(3^{d-1}+4^{d-1})\leq d4^d$ nearest neighbors , we have
\begin{align}\label{cFN}
|\mathcal{F}(r,a)|\leq (r-1)!(d4^d)^{r-1}\quad&\Longrightarrow\quad|\mathcal{F}(r)|\leq(L+1)^d(r-1)!(d4^d)^{r-1}
\\\nonumber&\Longrightarrow\quad|\mathcal{F}_N|\leq(L+1)^dN!(d4^d)^{N-1}.
\end{align}
Let $\mathcal{S}_N$ denote the event that the box $\Lambda_L$ and the subsets $\{\Upsilon(\Phi)\}_{\Phi\in\mathcal{F}_N}$ are all $L$-polynomially level spacing for $ H_{\varepsilon,\omega}$, using \eqref{probls}
and \eqref{cFN}, if $N=1$ we have
\beq\label{probs}
\P\{\mathcal{S}_N^c\}\leq Y_{\varepsilon_0}\left(1+(L+1)^dN!(d4^d)^{N-1}\right)(L+1)^{2d}(L+1)^{2d}L^{-(2\alpha-1)q}<\tfrac{1}{2}L^{-p}
\eeq
for sufficiently large $L$ since $p<(2\alpha-1)q-3d$.

Let $\mathcal{E}_N=\mathcal{B}_N\cap\mathcal{S}_N$. Combining \eqref{probb} and \eqref{probs}, we conclude that if $N=1$,
\beq
\P\{\mathcal{E}_N\}>1-\left((2Y)^{2d}P^2+\tfrac{1}{2}L^{-p}\right).
\eeq
To finish the proof we need to show that for all $\omega\in\mathcal{E}_N$ the box $\Lambda_L$ is $\theta$-PL for $H_{\varepsilon,\omega}$.

We fix $\omega\in\mathcal{E}_N$. Then we have \eqref{imploc2}, $\Lambda_L$ is polynomially level spacing for $H_{\varepsilon,\omega}$, and the subsets $\{\Upsilon_r\}_{r=1}^R$ constructed in \eqref{conups}
are $\theta$-PL-buffered subsets of $\Lambda_L$ for $H_{\varepsilon,\omega}$. It follows from \eqref{coverprop} and Definition~\ref{defbufall}(iii) that
\beq\label{Lamdecom}
\Lambda_L=\left\{\bigcup_{a\in\mathcal{G}}\Lambda_\ell^{\Lambda_L,\frac{\ell}{10}}(a)\right\}\cup\left\{\bigcup_{r=1}^R\Upsilon_r^{\Lambda_L,\frac{\ell}{10}}\right\}.
\eeq

We omit $\varepsilon$ and $\omega$ from the notation since they are now fixed. Let $\{(\psi_\lambda,\lambda)\}_{\lambda\in\sigma(H_{\Lambda_L})}$ be an eigensystem for $H_{\Lambda_L}$. For $a\in\mathcal{G}$,
let $\left\{(\varphi_x^{(a)}, \lambda_x^{(a)})\right\}_{x\in\Lambda_\ell(a)}$ be a $\theta$-polynomially localized eigensystem for $\Lambda_\ell(a)$. For $r=1,2,\ldots,R$, let
$\left\{(\phi_{\nu^{(r)}},\nu^{(r)})\right\}_{\nu^{(r)}\in\sigma(H_{\Upsilon_r})}$ be an eigensystem for $H_{\Upsilon_r}$, and set
\beq\label{sigmasigma}
\sigma_{\Upsilon_r}=\left\{\widetilde{\nu}^{(r)};\nu^{(r)}\in\sigma_{\mathcal{B}}(H_{\Upsilon_r})\right\}\subset  \sigma(H_{\Lambda_L})\setminus\sigma_{\mathcal{G}}(H_{\Lambda_L}),
\eeq
where $\widetilde{\nu}^{(r)}$ is given in \eqref{injectbad}, which also gives $\sigma_{\Upsilon_r}(H_{\Lambda_L})\subset\sigma(H_{\Lambda_L})\setminus\sigma_{\mathcal{G}_{\Upsilon_r}}(H_{\Lambda_L})$, but
the argument actually shows $\sigma_{\Upsilon_r}(H_{\Lambda_L})\subset\sigma(H_{\Lambda_L})\setminus\sigma_{\mathcal{G}}(H_{\Lambda_L})$. We also set
\beq
\sigma_{\mathcal{B}}(H_{\Lambda_L})=\bigcup_{r=1}^R\sigma_{\Upsilon_r}(H_{\Lambda_L})\subset  \sigma(H_{\Lambda_L})\setminus\sigma_{\mathcal{G}}(H_{\Lambda_L}).
\eeq

We claim
\beq\label{sigmaclaim}
\sigma(H_{\Lambda_L})=\sigma_{\mathcal{G}}(H_{\Lambda_L})\cup\sigma_{\mathcal{B}}(H_{\Lambda_L}).
\eeq
To do this, we assume $\lambda\in\sigma_{\mathcal{G}}\setminus(\sigma_{\mathcal{G}}(H_{\Lambda_L})\cup\sigma_{\mathcal{B}}(H_{\Lambda_L}))$. Since $\Lambda_L$ is polynomially level spacing for $H$,
Lemma~\ref{ideigsysall}(ii)(c) gives
\beq
|\psi_\lambda(y)|\leq C_{d,\varepsilon_0}L^q\ell^{-(\theta-2d)}\qtx{for all}y\in\bigcup_{a\in\mathcal{G}}\Lambda_\ell^{\Lambda_L,2\ell'}(a),
\eeq
and Lemma~\ref{lembadall}(ii) gives
\beq
|\psi_\lambda(y)|\leq C_{d,\varepsilon_0}L^{2d+2q}\ell^{-(\theta-2d)}\qtx{for all}y\in\bigcup_{r=1}^R\Upsilon_r^{\Lambda_L,2\ell'}.
\eeq
Using \eqref{Lamdecom} and $\theta-2d>\gamma_1\left(\tfrac{5d}{2}+2q\right)>\tfrac{5d}{2}+2q$, we conclude that
\beq
1=\|\psi_\lambda(y)\|\leq C_{d,\varepsilon_0}L^{2d+2q}\ell^{-(\theta-2d)}(L+1)^{\frac{d}{2}}<1
\eeq
for sufficiently large $\ell$, a contradiction. This establishes the claim.

We now index the eigenvalues and eigenvectors of $H_{\Lambda_L}$ by sites in $\Lambda_L$ using Hall's Marriage Theorem, which states a necessary and sufficient condition for the existence of
a perfect matching in a bipartite graph. (See \cite[Appendix~C]{EK} and \cite[Chapter~2]{BDM}.) We consider the bipartite graph $\mathbb{G}=(\Lambda_L,\sigma(H_{\Lambda_L});\mathbb{E})$, where
the edge set $\mathbb{E}\subset\Lambda_L\times\sigma(H_{\Lambda_L})$ is defined as follows. For each $\lambda\in\sigma_{\mathcal{G}}(H_{\Lambda_L})$ we fix
$\lambda_{x_\lambda}^{(a_\lambda)}\in\mathcal{E}_{\mathcal{G}}^{\Lambda_L}(\lambda)$, and set (recall \eqref{defUpsck} and \eqref{lsharp})
\beq
\mathcal{N}_0(x)=\begin{cases}
\{\lambda\in\sigma_{\mathcal{G}}(H_{\Lambda_L});\|x_\lambda-x\|<\ell_\sharp\}&\qtx{for}x\in\Lambda_L\setminus\bigcup_{r=1}^R\widehat{\Upsilon}_r
\\\emptyset&\qtx{for}x\in\bigcup_{r=1}^R\widehat{\Upsilon}_r
\end{cases}.
\eeq
We define
\beq\label{defN}
\mathcal{N}(x)=\begin{cases}
\mathcal{N}_0(x)&\qtx{for}x\in\Lambda_L\setminus\bigcup_{r=1}^R\widehat{\Upsilon'}_r
\\\sigma_\Upsilon(H_{\Lambda_L})&\qtx{for}x\in\widehat{\Upsilon}_r,\;r=1,2,\ldots,R
\\\mathcal{N}_0(x)\cup\sigma_\Upsilon(H_{\Lambda_L})&\qtx{for}x\in\widehat{\Upsilon'}_r,\setminus\widehat{\Upsilon}_r,\;r=1,2,\ldots,R
\end{cases},
\eeq
and let $\mathbb{E}=\{(x,\lambda)\in\Lambda_L\times\sigma(H_{\Lambda_L});\lambda\in\mathcal{N}(x)\}$.

$\mathcal{N}(x)$ was defined to ensure $|\psi_\lambda(x)|\ll 1 $ for $\lambda\not\in\mathcal{N}(x)$. This can be seen as follows:
\begin{itemize}

\item If $x\in\Lambda_L$ and $\lambda\in\sigma_{\mathcal{G}}(H_{\Lambda_L})\setminus\mathcal{N}_0(x)$, we have $\lambda=\widetilde{\lambda}_{x_\lambda}^{(a_\lambda)}$ with $\|x_\lambda-x\|\geq{\ell'}$,
so, using \eqref{defploc} and \eqref{diffun},
\begin{align}
|\psi_\lambda(x)|&\leq\left|\varphi_{x_\lambda}^{(a_\lambda)}(x)\right|+\left\|\varphi_{x_\lambda}^{(a_\lambda)}-\psi_\lambda\right\|\leq \ell^{-\Theta}
+2C_{d,\varepsilon_0}L^q\ell^{-\left(\theta-\frac{d-1}{2}\right)}
\\\nonumber&\leq3C_{d,\varepsilon_0}L^q\ell^{-\left(\theta-\frac{d-1}{2}\right)}.
\end{align}

\item If $x\in\Lambda_L\setminus\widehat{\Upsilon'}_r$ and $\lambda\in\sigma_{\Upsilon_r}(H_{\Lambda_L})$, then $\lambda=\widetilde{\nu}^{(r)}$ for some $\nu^{(r)}\in\sigma_{\mathcal{B}}(H_\Upsilon)$,
and, using \eqref{psidec5} and \eqref{diffun2}, (Note $\phi_{\nu^{(r)}}(x)=0$ if $x\not\in\Upsilon_r$.)
\begin{align}
|\psi_\lambda(x)|&\leq\left|\phi_{\nu^{(r)}}(x)\right|+\left\|\phi_{\nu^{(r)}}(x)-\psi_\lambda\right\|\leq C_{d,\varepsilon_0}L^q\ell^{-(\theta-2d)}+2C_{d,\varepsilon_0}L^{\frac{d}{2}+2q}\ell^{-(\theta-2d)}
\\\nonumber&\leq3C_{d,\varepsilon_0}L^{\frac{d}{2}+2q}\ell^{-(\theta-2d)}.
\end{align}
\end{itemize}
Therefore for all $x\in\Lambda_L$ and $\lambda\in\sigma(H_{\Lambda_L})\setminus\mathcal{N}(x)$ we have
\beq\label{psidec7}
|\psi_\lambda(x)|\leq C_{d,\varepsilon_0}L^{\frac{d}{2}+2q}\ell^{-(\theta-2d)}.
\eeq

Since $|\Lambda_L|=|\sigma(H_{\Lambda_L})|$, to apply Hall's Marriage Theorem we only need to verify $|\Theta|\leq|\mathcal{N}(\Theta)|$, where $\mathcal{N}(\Theta)=\bigcup_{x\in\Theta}\mathcal{N}(x)$
for $\Theta\subset\Lambda_L$. For $\Theta\subset\Lambda_L$, let $Q_\Theta$ be the orthogonal projection onto the span of $\{\psi_\lambda;\lambda\in\mathcal{N}(\Theta)\}$. If $\lambda\not\in\mathcal{N}(\Theta)$,
for all $x\in\Theta$ we have \eqref{psidec7}, thus
\begin{align}
\|(1-Q_\Theta)\chi_\Theta\|&\leq|\Lambda_L|^{\frac{1}{2}}|\Theta|^{\frac{1}{2}}C_{d,\varepsilon_0}L^{\frac{d}{2}+2q}\ell^{-(\theta-2d)}
\\\nonumber&\leq(L+1)^dC_{d,\varepsilon_0}L^{\frac{d}{2}+2q}\ell^{-(\theta-2d)}<1,
\end{align}
for sufficiently large $\ell$ since $\theta-2d>\gamma_1\left(\tfrac{5d}{2}+2q\right)>\tfrac{5}{2}d+2q$, so it follows from \cite[Lemma~A.1]{EK} that
\beq
|\Theta|=\tr\chi_\Theta\leq\tr Q_\Theta=|\mathcal{N}(\Theta)|.
\eeq

Using Hall's Marriage Theorem, we conclude that there exists a bijection
\beq
x\in\Lambda_L\mapsto\lambda_x\in\sigma(H_{\Lambda_L}),\qtx{where}\lambda_x\in\mathcal{N}(x).
\eeq
We set $\psi_x=\psi_{\lambda_x}$ for all $x\in\Lambda_L$.

To finish the proof we need to show that $\{(\psi_x,\lambda_x)\}_{x\in\Lambda_L}$ is a $\theta$-polynomially localized eigensystem for $\Lambda_L$.
We fix $N=1$, $x\in\Lambda_L$, take $y\in\Lambda_L$, and consider several cases:

\begin{enumerate}
\item Suppose $\lambda_x\in\sigma_{\mathcal{G}}(\Lambda_L)$. Then $x\in\Lambda_\ell(a_{\lambda_x})$ with $a_{\lambda_x}\in\mathcal{G}$, and $\lambda_x\in\sigma_{\{a_{\lambda_x}\}}(H_{\Lambda_L})$. In view of
\eqref{Lamdecom} we consider two cases:
\begin{enumerate}
\item If $y\in\Lambda_\ell^{\Lambda_L,\frac{\ell}{10}}(a)$ for some $a\in\mathcal{G}$ and $\|y-x\|\geq2\ell$, we must have $\Lambda_\ell(a_{\lambda_x})\cap\Lambda_\ell(a)=\emptyset$, so it follows from
\eqref{sigmaab} that $\lambda_x\not\in  \sigma_{\{a\}}(H_{\Lambda_L})$, and \eqref{psidec1} gives
    \beq\label{psigood}
    |\psi_x|\leq C_{d,\varepsilon_0}L^q\ell^{-(\theta-2d)}|\psi_x(y_1)|\sqtx{for some}y_1\in\partial^{\Theta,2\ell'}\Lambda_\ell(a).
    \eeq
\item If $y\in\Upsilon_1^{\Lambda_L,\frac{\ell}{10}}$, and $\|y-x\|\geq\ell+\diam\Upsilon_1$, we must have $\Lambda_\ell(a_{\lambda_x})\cap\Upsilon_1=\emptyset$, so it follows from \eqref{sigmaab} that
$\lambda_x\not\in\sigma_{\mathcal{G}_{\Upsilon_1}} (H_{\Lambda_L})$, and clearly $\lambda_x \not\in\sigma_{\Upsilon_1}(H_{\Lambda_L})$ in view of \eqref{sigmasigma}. Thus Lemma~\ref{lembadall}(ii) gives
    \beq\label{psibad}
    |\psi_x(y)|\leq C_{d,\varepsilon_0}L^{2d+2q}\ell^{-(\theta-2d)}|\psi_x(v)|\qtx{for some}v\in\partial^{\Lambda_L,2\ell'}\Upsilon_1.
    \eeq
\end{enumerate}
\item Suppose $\lambda_x\not\in\sigma_{\mathcal{G}}(\Lambda_L)$. Then it follows from \eqref{sigmaclaim} that we must have $\lambda_x\in    \sigma_{\Upsilon_1}(H_{\Lambda_L})$. If
$y\in\Lambda_{\ell}^{\Lambda_L,\frac{\ell}{10}}(a)$ for some $a\in\mathcal{G}$, and $\|y-x\|\geq\ell+\diam\Upsilon_1$, we must have $\Lambda_\ell(a)\cap\Upsilon_1=\emptyset$, and \eqref{psidec1} gives \eqref{psigood}.
\end{enumerate}

Now we fix $x\in\Lambda_L$, and take $y\in\Lambda_L $ such that $\|y-x\|\geq L'$. Suppose $|\psi_x(y)|>0$ without loss of generality.  We estimate $|\psi_x(y)| $ using either \eqref{psigood} or \eqref{psibad}
repeatedly, as appropriate, stopping when we get too close to $x$ so we are not in any case described above. (Note that this must happen since $|\psi_x(y)|>0$.) We accumulate decay only when using \eqref{psigood},
and just use $C_{d,\varepsilon_0}L^{2d+2q}\ell^{-(\theta-2d)}<1$ when using \eqref{psibad}, then recalling $L=Y\ell$, we get
\beq\label{repeatdec}
|\psi_x(y)|\leq\left(C_{d,\varepsilon_0}L^q\ell^{-(\theta-2d)}\right)^{n(Y)},
\eeq
where $n(Y)$ is the number of times we used \eqref{psigood}. We have
\beq
n(Y)(\ell+1)+\diam\Upsilon_1+2\ell\geq L'.
\eeq
Thus, using \eqref{sumdiam}, we have
\beq
n(Y)\geq\tfrac{1}{\ell+1}(L'-5\ell-2\ell)\geq\tfrac{\ell}{\ell+1}\left(\tfrac{Y}{40}-7\right)\geq2.
\eeq
for sufficiently large $\ell$ since $Y\geq400$. It follows from \eqref{repeatdec},
\beq
|\psi_x(y)|\leq\left(C_{d,\varepsilon_0}Y^q\ell^{-(\theta-2d-q)}\right)^2\leq L^{-\theta},
\eeq
for sufficiently large $\ell$ since $2(\theta-2d-q)=\theta+(\theta-4d-2q)>\theta$.

We conclude that $\{(\psi_x,\lambda_x)\}_{x\in\Lambda_L}$ is a $\theta$-polynomially localized eigensystem for $\Lambda_L$, so the box $\Lambda_L$ is $\theta$-polynomially localizing for $ H_{\varepsilon,\omega}$.

\end{proof}

\begin{proof}[Proof of Proposition~\ref{propmsa1}]
We assume \eqref{initcon} and set $L_{k+1}=YL_k$ for $k=0,1,\ldots$. We set
\beq
P_k=\sup_{x\in\R^d}\P\{\Lambda_{L_k}(x)\sqtx{is not}\theta\text{-polynomially localizing for}\;H_{\varepsilon,\omega}\}\sqtx{for}k=1,2,\ldots.
\eeq
Then by Lemma~\ref{indumsa1}, we have
\beq\label{calpk}
P_{k+1}\leq(2Y)^{2d}P_k^2+\tfrac{1}{2}L_{k+1}^{-p}\qtx{for}k=0,1,\ldots
\eeq
If $P_k\leq L_k^{-p}$ for some $k\geq0$, we have
\beq
P_{k+1}\leq(2Y)^{2d}L_k^{-2p}+\tfrac{1}{2}L_{k+1}^{-p}\leq(2Y)^{2d+2p}L_{k+1}^{-2p}+\tfrac{1}{2}L_{k+1}^{-p}\leq L_{k+1}^{-p}
\eeq
for $L_0$ sufficiently large. Therefore to finish the proof, we need to show that
\beq
K_0=\inf\{k\in\N;P_k\leq L_k^{-p}\}<\infty.
\eeq

It follows from \eqref{calpk} that for any $1\leq k<K_0$,
\beq
P_k\leq(2Y)^{2d}P_{k-1}^2+\tfrac{1}{2}L_k^{-p}<(2Y)^{2d}P_{k-1}^2+\tfrac{1}{2}P_k,
\eeq
so
\beq
2(2Y)^{2d}P_k<\left(2(2Y)^{2d}P_{k-1}\right)^2.
\eeq
Therefore for $1\leq k<K_0$, we have
\beq\label{ineqpk}
2^{2d+1}Y^{-(kp-2d)}L_0^{-p}=2(2Y)^{2d}L_k^{-p}<2(2Y)^{2d}P_k<\left(2(2Y)^{2d}P_0\right)^{2^k}.
\eeq
Since $2(2Y)^{2d}P_0<1$, \eqref{ineqpk} cannot be satisfied for large $k$. We conclude that $K_0<\infty$.
\end{proof}

\subsection{The first intermediate step}

\begin{proposition}\label{indumsa1one}
Fix $\varepsilon_0>0$. Suppose for some scale $\ell$ and $0<\varepsilon\leq\varepsilon_0$ we have
\beq\label{hypmsaindone}
\inf_{x\in\R^d}\P\{\Lambda_\ell(x)\sqtx{is}\theta\text{-polynomially localizing for}\;H_{\varepsilon,\omega}\}\geq1-\ell^{-p}.
\eeq
Then, if $\ell$ is sufficiently large, for $L=\ell^{\gamma_1}$ we have
\beq\label{resone}
\inf_{x\in\R^d}\P\{\Lambda_L(x)\sqtx{is}m_0^\ast\text{-mix localizing for}\;H_{\varepsilon,\omega}\}\geq1-L^{-p},
\eeq
where
\beq\label{mst000}
m_0^\ast\geq\tfrac{1}{8}\left(\tfrac{5d}{2}+q\right)L^{-(1-\tau+\frac{1}{\gamma_1})}\log L.
\eeq
\end{proposition}

\begin{proof}
We follow the proof of Lemma~\ref{indumsa1}. For $N\in\N$, let $\mathcal{B}_N$, $\mathcal{S}_N$ and $\mathcal{E}_N$ as in the proof of Lemma~\ref{indumsa1}. Using \eqref{hypmsaindone}, \eqref{covernum} and the
fact that events on disjoint boxes are independent, if $N=1$ we have,
\beq\label{probbone}
\P\{\mathcal{B}_N^c\}\leq\left(\tfrac{2L}{\ell}\right)^{2d}\ell^{-2p}=2^{2d}\ell^{-2p-2d(\gamma_1-1)}<\tfrac{1}{2}\ell^{-\gamma_1p}=\tfrac{1}{2}L^{-p}
\eeq
for all $\ell$ sufficiently large since $1<\gamma_1<1+\tfrac{p}{p+2d}$. Also, using \eqref{probls} and \eqref{cFN}, if $N=1$ we have,
\beq\label{probsone}
\P\{\mathcal{S}_N^c\}\leq\left(1+(L+1)^d\right)Y_{\varepsilon_0}(L+1)^{2d}L^{-(2\alpha-1)q}<\tfrac{1}{2}L^{-p}
\eeq
for sufficiently large $L$, since $p<(2\alpha-1)q-3d$. Combining \eqref{probbone} and \eqref{probsone}, we conclude that
\beq
\P\{\mathcal{E}_N\}>1-L^{-p}.
\eeq

To finish the proof we need to show that for all $\omega\in\mathcal{E}_N$ the box $\Lambda_L$ is $m_0^\ast$-mix localizing for $H_{\varepsilon,\omega}$, where $m_0^\ast$ is given in \eqref{mst000}. Following
the proof of Lemma~\ref{indumsa1}, we get \eqref{sigmaclaim} and obtain an eigensystem $\{(\psi_x,\lambda_x)\}_{x\in\Lambda_L}$ for $H_{\Lambda_L}$ using Hall's Marriage Theorem. To finish the proof we
need to show that $\{(\psi_x,\lambda_x)\}_{x\in\Lambda_L}$ is an $m_0^\ast$-localized eigensystem for $\Lambda_L$. We proceed as in the proof of Lemma~\ref{indumsa1}. We fix $N=1$, $x\in\Lambda_L$, and
take $y\in\Lambda_L$ such that $\|y-x\|\geq L_\tau$, we have
\beq
n(\ell)(\ell+1)+\diam\Upsilon_1+2\ell\geq L_\tau.
\eeq
where $n(\ell)$ is the number of times we used \eqref{psigood}.
Thus, using \eqref{sumdiam}, we have
\beq
n(\ell)\geq\tfrac{1}{\ell+1}(L_\tau-5\ell-2\ell)\geq\tfrac{\ell}{\ell+1}\left(\tfrac{1}{2}\ell^{\gamma_1\tau-1}-7\right)\geq\tfrac{1}{4}\ell^{\gamma_1\tau-1}.
\eeq
for sufficiently large $\ell$. It follows from \eqref{repeatdec},
\begin{align}
|\psi_x(y)|&\leq\left(C_{d,\varepsilon_0}\ell^{-(\theta-2d-\gamma_1q)}\right)^{\frac{1}{4}\ell^{\gamma_1\tau-1}}
\\\nonumber&\leq\e^{-\frac{1}{8}\left(\frac{5d}{2}+q\right)L^{-(1-\tau+\frac{1}{\gamma_1})}(\log L)\|y-x\|},
\end{align}
for sufficiently large $\ell$.

We conclude that $\{(\psi_x,\lambda_x)\}_{x\in\Lambda_L}$ is an $m_0^\ast$-localized eigensystem for $\Lambda_L$, where $m_0^\ast$ is given in \eqref{mst000}, so the box $\Lambda_L$ is
$m_0^\ast$-mix localizing for $ H_{\varepsilon,\omega}$.

\end{proof}

\subsection{The second multiscale analysis}

\begin{proposition}\label{propmsa2}
Fix $\varepsilon_0>0$. There exists a finite scale $\mathcal{L}(\varepsilon_0)$ with the following property: Suppose for some scale $L_0\geq\mathcal{L}(\varepsilon_0)$,
$0<\varepsilon\leq\varepsilon_0$, and $m^\ast_0\geq L_0^{-\kappa}$ where $0<\kappa<\tau$, we have
\beq\label{initcon2}
\inf_{x\in\R^d}\P\{\Lambda_{L_0}(x)\sqtx{is}m^\ast_0\text{-mix localizing for}\;H_{\varepsilon,\omega}\}\geq1-L_0^{-p}.
\eeq
Then, setting $L_{k+1}=L_k^{\gamma_1}$ for $k=0,1,\ldots$, we have
\beq\label{res2}
\inf_{x\in\R^d}\P\{\Lambda_{L_k}(x)\sqtx{is}\tfrac{m^\ast_0}{2}\text{-mix localizing for}\;H_{\varepsilon,\omega}\}\geq1-L_k^{-p}\sqtx{for}k=0,1,\ldots.
\eeq
\end{proposition}

Proposition~\ref{propmsa2} follows from the following induction step for the multiscale analysis.

\begin{lemma}\label{indumsa2}
Fix $\varepsilon_0>0$. Suppose for some scale $\ell$, $0<\varepsilon\leq\varepsilon_0$, and $m^\ast\geq\ell^{-\kappa}$, where $0<\kappa<\tau$, we have
\beq\label{hypmsaind2}
\inf_{x\in\R^d}\P\{\Lambda_\ell(x)\sqtx{is}m^\ast\text{-mix localizing for}\;H_{\varepsilon,\omega}\}\geq1-\ell^{-p}.
\eeq
Then, if $\ell$ is sufficiently large, for $L=\ell^{\gamma_1}$ we have
\beq
\inf_{x\in\R^d}\P\{\Lambda_L(x)\sqtx{is}M^\ast\text{-mix localizing for}\;H_{\varepsilon,\omega}\}\geq1-L^{-p},
\eeq
where
\beq\label{bigMst}
M^\ast\geq m^\ast\left(1-C_{d,\varepsilon_0}\gamma_1q\ell^{-\min\left\{\frac{1-\tau}{2},\gamma_1\tau-1,\tau-\kappa\right\}}\right)\geq L^{-\kappa}.
\eeq
\end{lemma}

\begin{proof}
We follow the proof of Lemma~\ref{indumsa1}. For $N\in\N$, let $\mathcal{B}_N$ denote the event that there do not exist two disjoint boxes in $\mathcal{C}_{L,\ell}$ that are not
$m^\ast$-mix localizing for $H_{\varepsilon,\omega}$. Using \eqref{hypmsaind2}, \eqref{covernum} and the fact that events on disjoint boxes are independent, if $N=1$ we have
\beq\label{probb2}
\P\{\mathcal{B}_N^c\}\leq\left(\tfrac{2L}{\ell}\right)^{(N+1)d}\ell^{-(N+1)p}=2^{2d}\ell^{-(2p-2d(\gamma_1-1))}<\tfrac{1}{2}\ell^{-\gamma_1p}=\tfrac{1}{2}L^{-p}
\eeq
for all $\ell$ sufficiently large since $1<\gamma_1<1+\tfrac{p}{p+2d}$.

We now fix $\omega\in\mathcal{B}_N$, and proceed as in the proof of Lemma~\ref{indumsa1} with $\sharp$ being $m^\ast$-ML.
Then we have $\Upsilon_r$, $r=1,2,\ldots,R$ such that each $\Upsilon_r$ satisfies all the requirements to be an $m^\ast$-ML-buffered subset of $\Lambda_L$ with
$\mathcal{G}_{\Upsilon_r}=\partial_{\mathrm{ex}}^{\mathbb{G}_1}\widetilde{\Phi}_r$, except we do not know if $\Upsilon_r$ is $L$-polynomially level spacing for $H_{\varepsilon,\omega}$.

Given a $\mathbb{G}_2$-connected subset $\Phi$ of $\Xi_{L,\ell}$, let $\Upsilon(\Phi)\subset\Lambda_L$ be constructed from $\Phi$ as in \eqref{conups} with $\sharp$ being $m^\ast$-ML.
Let $\mathcal{S}_N$ denote the event that the box $\Lambda_L$ and the subsets $\{\Upsilon(\Phi)\}_{\Phi\in\mathcal{F}_N}$ are all $L$-polynomially level spacing for $H_{\varepsilon,\omega}$.
Using \eqref{probls} and \eqref{cFN}, if $N=1$ we have
\beq\label{probs2}
\P\{\mathcal{S}^c\}\leq\left(1+\left(\tfrac{2L}{\ell}\right)^d\right)Y_{\varepsilon_0}(L+1)^{2d}L^{-(2\alpha-1)q}<\tfrac{1}{2}L^{-p}
\eeq
for sufficiently large $L$, since $p<(2\alpha-1)q-3d$.

Let $\mathcal{E}_N=\mathcal{B}_N\cap\mathcal{S}_N$. Combining \eqref{probb2} and \eqref{probs2}, we conclude that if $N=1$,
\beq
\P\{\mathcal{E}_N\}>1-L^{-p}.
\eeq
To finish the proof we need to show that for all $\omega\in\mathcal{E}_N$ the box $\Lambda_L$ is $M^\ast$-mix localizing for $H_{\varepsilon,\omega}$, where $M^\ast$ is
given in \eqref{bigMst}.

We fix $\omega\in\mathcal{E}_N$. Then we have \eqref{imploc2}, $\Lambda_L$ is polynomially level spacing for $H_{\varepsilon,\omega}$, and the subsets $\{\Upsilon_r\}_{r=1}^R$ constructed
in \eqref{conups} are $m^\ast$-ML-buffered subset of $\Lambda_L$ for $H_{\varepsilon,\omega}$. We proceed as in the proof of Lemma~\ref{indumsa1}. To claim \eqref{sigmaclaim}, we
assume $\lambda\in\sigma_{\mathcal{G}}\setminus(\sigma_{\mathcal{G}}(H_{\Lambda_L})\cup\sigma_{\mathcal{B}}(H_{\Lambda_L}))$. Since $\Lambda_L$ is polynomially level spacing for $H$,
Lemma~\ref{ideigsysall}(ii)(c) gives
\beq
|\psi_\lambda(y)|\leq\e^{-m^\ast_2\ell_\tau}\qtx{for all}y\in\bigcup_{a\in\mathcal{G}}\Lambda_\ell^{\Lambda_L,2\ell_\tau}(a),
\eeq
and Lemma~\ref{lembadall}(ii) gives
\beq
|\psi_\lambda(y)|\leq\e^{-m^\ast_5\ell_\tau}\qtx{for all}y\in\bigcup_{r=1}^R\Upsilon_r^{\Lambda_L,2\ell_\tau}.
\eeq
Using \eqref{Lamdecom}, we conclude that (note $m^\ast_5\leq m^\ast_2$)
\beq
1=\|\psi_\lambda(y)\|\leq\e^{-m^\ast_5\ell_\tau}(L+1)^{\frac{d}{2}}<1,
\eeq
 a contradiction. This establishes the claim.

To index the eigenvalues and eigenvectors of $H_{\Lambda_L}$ by sites in $\Lambda_L$, we define $\mathcal{N}(x)$ as in \eqref{defN} and proceed as in the proof of Lemma~\ref{indumsa1}. We have:
\begin{itemize}
\item If $x\in\Lambda_L$ and $\lambda\in\sigma_{\mathcal{G}}(H_{\Lambda_L})\setminus\mathcal{N}_0(x)$, we have $\lambda=\widetilde{\lambda}_{x_\lambda}^{(a_\lambda)}$ with
$\|x_\lambda-x\|\geq{\ell_\tau}$, so, using \eqref{defloc} and \eqref{diffunstr},
\beq
|\psi_\lambda(x)|\leq\left|\varphi_{x_\lambda}^{(a_\lambda)}(x)\right|+\left\|\varphi_{x_\lambda}^{(a_\lambda)}-\psi_\lambda\right\|\leq \e^{-m^\ast\ell_\tau}+2\e^{-m^\ast_1\ell_\tau}L^q\leq3\e^{-m_1\ell_\tau}L^q.
\eeq

\item If $x\in\Lambda_L\setminus\widehat{\Upsilon'}_r$ and $\lambda\in\sigma_{\Upsilon_r}(H_{\Lambda_L})$, then $\lambda=\widetilde{\nu}^{(r)}$ for some
$\nu^{(r)}\in\sigma_{\mathcal{B}}(H_{\Upsilon_r})$, and, using \eqref{psidec5} and \eqref{diffun2str}, (Note $\phi_{\nu^{(r)}}(x)=0$ if $x\not\in\Upsilon_r$.)
\beq
|\psi_\lambda(x)|\leq\left|\phi_{\nu^{(r)}}(x)\right|+\left\|\phi_{\nu^{(r)}}(x)-\psi_\lambda\right\|\leq \e^{-m^\ast_2\ell_\tau}+2\e^{-m^\ast_4\ell_\tau}L^q\leq3\e^{-m^\ast_4\ell_\tau}L^q.
\eeq
\end{itemize}
Therefore for all $x\in\Lambda_L$ and $\lambda\in\sigma(H_{\Lambda_L})\setminus\mathcal{N}(x)$ we have
\beq\label{psidec7str}
|\psi_\lambda(x)|\leq3\e^{-m^\ast_4\ell_\tau}L^q\leq\e^{-\frac{1}{2}m^\ast_4\ell_\tau}.
\eeq
If $\lambda\not\in\mathcal{N}(\Theta)$, for all $x\in\Theta$ we have \eqref{psidec7str}, thus
\beq
\|(1-Q_\Theta)\chi_\Theta\|\leq|\Lambda_L|^{\frac{1}{2}}|\Theta|^{\frac{1}{2}}\e^{-\frac{1}{2}m^\ast_4\ell_\tau}\leq(L+1)^d\e^{-\frac{1}{2}m^\ast_4\ell_\tau}<1.
\eeq
Following the proof of Lemma~\ref{indumsa1}, we can apply Hall's Marriage Theorem to obtain an eigensystem $\{(\psi_x,\lambda_x)\}_{x\in\Lambda_L}$ for $H_{\Lambda_L}$.

To finish the proof we need to show that $\{(\psi_x,\lambda_x)\}_{x\in\Lambda_L}$ is an $M^\ast$-localized eigensystem for $\Lambda_L$, where $M^\ast$ is given in \eqref{bigMst}.
We fix $N=1$, $x\in\Lambda_L$, take $y\in\Lambda_L$, and consider several cases:

\begin{enumerate}
\item[(i)] Suppose $\lambda_x\in\sigma_{\mathcal{G}}(\Lambda_L)$. Then $x\in\Lambda_\ell(a_{\lambda_x})$ with $a_{\lambda_x}\in\mathcal{G}$, and
$\lambda_x\in\sigma_{\{a_{\lambda_x}\}}(H_{\Lambda_L})$. In view of \eqref{Lamdecom} we consider two cases:
\begin{enumerate}
\item If $y\in\Lambda_\ell^{\Lambda_L,\frac{\ell}{10}}(a)$ for some $a\in\mathcal{G}$ and $\|y-x\|\geq2\ell$, we must have $\Lambda_\ell(a_{\lambda_x})\cap\Lambda_\ell(a)=\emptyset$,
so it follows from \eqref{sigmaab} that $\lambda_x\not\in  \sigma_{\{a\}}(H_{\Lambda_L})$, and \eqref{pdec01} gives
    \beq\label{psigoodstr}
    |\psi_x|\leq\e^{-m^\ast_3\|y_1-y\|}|\psi_x(y_1)|\sqtx{for some}y_1\in\partial^{\Theta,\ell_{\widetilde{\tau}}}\Lambda_\ell(a).
    \eeq
\item If $y\in\Upsilon_1^{\Lambda_L,\frac{\ell}{10}}$, and $\|y-x\|\geq\ell+\diam\Upsilon_1$, we must have $\Lambda_\ell(a_{\lambda_x})\cap\Upsilon_1=\emptyset$, so it
follows from \eqref{sigmaab} that $\lambda_x\not\in\sigma_{\mathcal{G}_{\Upsilon_1}} (H_{\Lambda_L})$, and clearly $\lambda_x \not\in\sigma_{\Upsilon_1}(H_{\Lambda_L})$ in view of
\eqref{sigmasigma}. Thus Lemma~\ref{lembadall}(ii) gives
    \beq\label{psibadstr}
    |\psi_x(y)|\leq\e^{-m^\ast_5\ell_\tau}|\psi_x(v)|\qtx{for some}v\in\partial^{\Lambda_L,2\ell_\tau}\Upsilon_1.
    \eeq
\end{enumerate}
\item Suppose $\lambda_x\not\in\sigma_{\mathcal{G}}(\Lambda_L)$. Then it follows from \eqref{sigmaclaim} that we must have $\lambda_x\in    \sigma_{\Upsilon_1}(H_{\Lambda_L})$. If
$y\in\Lambda_{\ell}^{\Lambda_L,\frac{\ell}{10}}(a)$ for some $a\in\mathcal{G}$, and $\|y-x\|\geq\ell+\diam\Upsilon_1$, we must have $\Lambda_\ell(a)\cap\Upsilon_1=\emptyset$, and
\eqref{pdec01} gives \eqref{psigoodstr}.
\end{enumerate}

Now we fix $x\in\Lambda_L$, and take $y\in\Lambda_L $ such that $\|y-x\|\geq L_\tau$. Suppose $|\psi_x(y)|>0$ without loss of generality.  We estimate $|\psi_x(y)| $ using either
\eqref{psigoodstr} or \eqref{psibadstr} repeatedly, as appropriate, stopping when we get too close to $x$ so we are not in any case described above. (Note that this must happen since
$|\psi_x(y)|>0$.) We accumulate decay only when using \eqref{psigoodstr}, and just use $\e^{-m^\ast_5\ell_\tau}<1$ when using \eqref{psibadstr}, then we get
\begin{align}
|\psi_x(y)|&\leq\e^{-m^\ast_3\left(\|y-x\|-\diam\Upsilon-2\ell\right)}\leq\e^{-m^\ast_3\left(\|y-x\|-7\ell\right)}
\\\nonumber&\leq\e^{-m^\ast_3\|y-x\|\left(1-7\ell^{1-\gamma_1\tau}\right)}\leq\e^{M\|y-x\|},
\end{align}
where we used \eqref{sumdiam} and took
\begin{align}
M^\ast&=m^\ast_3\left(1-7\ell^{1-\gamma_1\tau}\right)\geq \left(m^\ast\left(1-4\ell^{\frac{\tau-1}{2}}\right)-C_{d,\varepsilon_0}\gamma_1q\tfrac{\log\ell}{\ell_{\widetilde{\tau}}}\right)\left(1-7\ell^{1-\gamma_1\tau}\right)
\\\nonumber&\geq m^\ast\left(1-4\ell^{\frac{\tau-1}{2}}-C_{d,\varepsilon_0}\gamma_1q\ell^{\kappa-\tau}\right)\left(1-7\ell^{1-\gamma_1\tau}\right)
\\\nonumber&\geq m^\ast\left(1-C_{d,\varepsilon_0}\gamma_1q\ell^{-\min\{\frac{1-\tau}{2},\gamma_1\tau-1,\tau-\kappa\}}\right)
\\\nonumber&\geq\tfrac{1}{2}\ell^{-\kappa}\geq\ell^{-\gamma_1\kappa}=L^{-\kappa}
\end{align}
for $\ell$ sufficiently large, where we used \eqref{mst3} and $m^\ast\geq\ell^{-\kappa}$.

We conclude that $\{(\psi_x,\lambda_x)\}_{x\in\Lambda_L}$ is an $M^\ast$-localized eigensystem for $\Lambda_L$, where $M^\ast$ is given in \eqref{bigMst}, so the box $\Lambda_L$
is $M^\ast$-mix localizing for $ H_{\varepsilon,\omega}$.

\end{proof}

\begin{proof}[Proof of Proposition~\ref{propmsa2}]
We assume \eqref{initcon2} and set $L_{k+1}=L_k^{\gamma_1}$ for $k=0,1,\ldots$. If $L_0$ is sufficiently large it follows from Lemma~\ref{indumsa2} by an induction argument that
\beq
\inf_{x\in\R^d}\P\{\Lambda_{L_k}(x)\sqtx{is}m^\ast_k\text{-localizing for}\;H_{\varepsilon,\omega}\}\geq1-L_k^{-p}\sqtx{for}k=0,1,\ldots,
\eeq
where for $k=1,2,\ldots$ we have
\beq
m^\ast_k\geq m^\ast_{k-1}\left(1-C_{d,\varepsilon_0}\gamma_1qL_{k-1}^{-\varrho}\right),\sqtx{with}\varrho=\min\left\{\tfrac{1-\tau}{2},\gamma_1\tau-1,\tau-\kappa\right\}.
\eeq
Thus for all $k=1,2,\ldots$, taking $L_0$ sufficiently large we get
\beq
m^\ast_k\geq m^\ast_0\prod_{j=0}^{k-1}\left(1-C_{d,\varepsilon_0}\gamma_1qL_0^{-\varrho\gamma^j}\right)\geq m^\ast_0\prod_{j=0}^{\infty}\left(1-C_{d,\varepsilon_0}\gamma_1qL_0^{-\varrho\gamma_1^j}\right)\geq\tfrac{m^\ast_0}{2},
\eeq
finishing the proof of Proposition~\ref{propmsa2}.
\end{proof}

\subsection{The third multiscale analysis}

\begin{proposition}\label{propmsa3}
Fix $\varepsilon_0>0$, $Y\geq400^{\frac{1}{1-s}}$, and $\widetilde{P}_0<\left(2(2Y)^{(\lfloor Y^s\rfloor+1)d}\right)^{-\frac{1}{\lfloor Y^s\rfloor}}$. There exists a finite scale $\mathcal{L}(\varepsilon_0,Y)$
with the following property: Suppose for some scale $L_0\geq\mathcal{L}(\varepsilon_0,Y)$ and $0<\varepsilon\leq\varepsilon_0$ we have
\beq\label{initcon3}
\inf_{x\in\R^d}\P\{\Lambda_{L_0}(x)\sqtx{is}s\text{-SEL for}\;H_{\varepsilon,\omega}\}\geq1-\widetilde{P}_0.
\eeq
Then, setting $L_{k+1}=YL_k$ for $k=0,1,\ldots$, there exists $K_0=K_0(Y,L_0,\widetilde{P}_0)\in\N$ such that
\beq\label{res3}
\inf_{x\in\R^d}\P\{\Lambda_{L_k}(x)\sqtx{is}s\text{-SEL for}\;H_{\varepsilon,\omega}\}\geq1-\e^{-L_k^\zeta}\sqtx{for}k\geq K_0.
\eeq
\end{proposition}

Proposition~\ref{propmsa3} follows from the following induction step for the multiscale analysis.

\begin{lemma}\label{indumsa3}
Fix $\varepsilon_0>0$, $Y\geq400^{\tfrac{1}{1-s}}$, and $0\le P\leq1$. Suppose for some scale $\ell$ and $0<\varepsilon\leq\varepsilon_0$ we have
\beq\label{hypmsaind3}
\inf_{x\in\R^d}\P\{\Lambda_\ell(x)\sqtx{is}s\text{-SEL for}\;H_{\varepsilon,\omega}\}\geq1-P.
\eeq
Then, if $\ell$ is sufficiently large, for $L=Y\ell$ we have
\beq
\inf_{x\in\R^d}\P\{\Lambda_L(x)\sqtx{is}s\text{-SEL for}\;H_{\varepsilon,\omega}\}\geq1-\left((2Y)^{(\lfloor Y^s\rfloor+1)d}P^{\lfloor Y^s\rfloor+1}+\tfrac{1}{2}\e^{-L^\zeta}\right).
\eeq
\end{lemma}

\begin{proof}
We follow the proof of Lemma~\ref{indumsa1}. For $N\in\N$, let $\mathcal{B}_N$ denote the event that there exist at most $N$ disjoint boxes in $\mathcal{C}_{L,\ell}$ that are
not $s$-SEL for $H_{\varepsilon,\omega}$. Using \eqref{hypmsaind3}, \eqref{covernum} and the fact that events on disjoint boxes are independent, if $N=\lfloor Y^s\rfloor$ we have
\beq\label{probbsub}
\P\{\mathcal{B}^c\}\leq\left(\tfrac{2L}{\ell}\right)^{(N+1)d}P^{N+1}=(2Y)^{\left(\lfloor Y^s\rfloor+1\right)d}P^{\lfloor Y^s\rfloor+1}.
\eeq

We now fix $\omega\in\mathcal{B}_N$, and proceed as in the proof of Lemma~\ref{indumsa1} with $\sharp$ being $s$-SEL.
Then we have $\Upsilon_r$, $r=1,2,\ldots,R$ such that each $\Upsilon_r$ satisfies all the requirements to be an $s$-SEL-buffered subset of $\Lambda_L$ with $\mathcal{G}_{\Upsilon_r}
=\partial_{\mathrm{ex}}^{\mathbb{G}_1}\widetilde{\Phi}_r$, except we do not know if $\Upsilon_r$ is $L$-level spacing for $H_{\varepsilon,\omega}$.

It follows from Lemma~\ref{lemlsp} that for any $\Theta\subset\Lambda_L$ we have
\begin{align}\label{problssub}
\P\{\Theta\sqtx{is}L\text{-level spacing for}\;H_{\varepsilon,\omega}\}\geq1-Y_{\varepsilon_0}\e^{-(2\alpha-1)L^\beta}(L+1)^{2d}.
\end{align}
Given a $\mathbb{G}_2$-connected subset $\Phi$ of $\Xi_{L,\ell}$, let $\Upsilon(\Phi)\subset\Lambda_L$ be constructed from $\Phi$ as in \eqref{conups} with $\sharp$ being
$s$-SEL. Let $\mathcal{S}_N$ denote the event that the box $\Lambda_L$ and the subsets the subsets $\{\Upsilon(\Phi)\}_{\Phi\in\mathcal{F}_N}$ are all $L$-level spacing for
$H_{\varepsilon,\omega}$. Using \eqref{problssub} and \eqref{cFN}, if $N=\lfloor Y^s\rfloor$ we have
\beq\label{probssub}
\P\{\mathcal{S}_N^c\}\leq Y_{\varepsilon_0}\left(1+(L+1)^dN!(d4^d)^{N-1}\right)(L+1)^{2d}\e^{-(2\alpha-1)L^\beta}<\tfrac{1}{2}\e^{-L^\zeta}
\eeq
for sufficiently large $L$, since $\zeta<\beta$.

Let $\mathcal{E_N}=\mathcal{B_N}\cap\mathcal{S_N}$. Combining \eqref{probbsub} and \eqref{probssub}, we conclude that
\beq
\P\{\mathcal{E}_N\}>1-\left((2Y)^{\left(\lfloor Y^s\rfloor+1\right)d}P^{\lfloor Y^s\rfloor+1}+\tfrac{1}{2}\e^{-L^\zeta}\right).
\eeq
To finish the proof we need to show that for all $\omega\in\mathcal{E}_N$ the box $\Lambda_L$ is $s$-SEL for $H_{\varepsilon,\omega}$.

We fix $\omega\in\mathcal{E}_N$. Then we have \eqref{imploc2}, $\Lambda_L$ is level spacing for $H_{\varepsilon,\omega}$, and the subsets $\{\Upsilon_r\}_{r=1}^R$ constructed
in \eqref{conups} are $s$-SEL-buffered subsets of $\Lambda_L$ for $H_{\varepsilon,\omega}$. We proceed as in the proof of Lemma~\ref{indumsa1}. To claim \eqref{sigmaclaim},
we assume $\lambda\in\sigma_{\mathcal{G}}\setminus(\sigma_{\mathcal{G}}(H_{\Lambda_L})\cup\sigma_{\mathcal{B}}(H_{\Lambda_L}))$. Since $\Lambda_L$ is level spacing for $H$,
Lemma~\ref{ideigsysall}(ii)(c) gives
\beq
|\psi_\lambda(y)|\leq\e^{-c_2\ell^s}\qtx{for all}y\in\bigcup_{a\in\mathcal{G}}\Lambda_\ell^{\Lambda_L,2\ell'}(a),
\eeq
and Lemma~\ref{lembadall}(ii) gives
\beq
|\psi_\lambda(y)|\leq\e^{-c_4\ell^s}\qtx{for all}y\in\bigcup_{r=1}^R\Upsilon_r^{\Lambda_L,2\ell'}.
\eeq
Using \eqref{Lamdecom}, we conclude that (note $c_4\leq c_2$)
\beq
1=\|\psi_\lambda(y)\|\leq\e^{-c_4\ell^s}(L+1)^{\frac{d}{2}}<1,
\eeq
a contradiction. This establishes the claim.

To index the eigenvalues and eigenvectors of $H_{\Lambda_L}$ by sites in $\Lambda_L$, we define $\mathcal{N}(x)$ as in \eqref{defN} proceed as in the proof of
Lemma~\ref{indumsa1}. We have:
\begin{itemize}

\item If $x\in\Lambda_L$ and $\lambda\in\sigma_{\mathcal{G}}(H_{\Lambda_L})\setminus\mathcal{N}_0(x)$, we have $\lambda=\widetilde{\lambda}_{x_\lambda}^{(a_\lambda)}$
with $\|x_\lambda-x\|\geq{\ell'}$, so, using \eqref{defsloc} and \eqref{diffunsub},
\beq
|\psi_\lambda(x)|\leq\left|\varphi_{x_\lambda}^{(a_\lambda)}(x)\right|+\left\|\varphi_{x_\lambda}^{(a_\lambda)}-\psi_\lambda\right\|\leq \e^{-\ell^s}
+2\e^{-c_1\ell^s}\e^{L^\beta}\leq3\e^{-c_1\ell^s}\e^{L^\beta}.
\eeq

\item If $x\in\Lambda_L\setminus\widehat{\Upsilon'}_r$ and $\lambda\in\sigma_{\Upsilon_r}(H_{\Lambda_L})$, then $\lambda=\widetilde{\nu}^{(r)}$ for some
$\nu^{(r)}\in\sigma_{\mathcal{B}}(H_{\Upsilon_r})$, and, using \eqref{psidec5} and \eqref{diffun2sub}, (Note $\phi_{\nu^{(r)}}(x)=0$ if $x\not\in\Upsilon_r$.)
\beq
|\psi_\lambda(x)|\leq\left|\phi_\nu(x)\right|+\left\|\phi_\nu(x)-\psi_\lambda\right\|\leq \e^{-c_2\ell^s}+2\e^{-c_3\ell^s}\e^{L^\beta}\leq3\e^{-c_3\ell^s}\e^{L^\beta}.
\eeq
\end{itemize}
Therefore for all $x\in\Lambda_L$ and $\lambda\in\sigma(H_{\Lambda_L})\setminus\mathcal{N}(x)$ we have
\beq\label{psidec7sub}
|\psi_\lambda(x)|\leq3\e^{-c_3\ell^s}\e^{L^\beta}\leq\e^{-\frac{1}{2}c_3\ell^s}.
\eeq
If $\lambda\not\in\mathcal{N}(\Theta)$, for all $x\in\Theta$ we have \eqref{psidec7sub}, thus
\beq
\|(1-Q_\Theta)\chi_\Theta\|\leq|\Lambda_L|^{\frac{1}{2}}|\Theta|^{\frac{1}{2}}\e^{-\frac{1}{2}c_3\ell^s}\leq(L+1)^d\e^{-\frac{1}{2}c_3\ell^s}<1.
\eeq
Following the proof of Lemma~\ref{indumsa1}, we can apply Hall's Marriage Theorem to obtain an eigensystem $\{(\psi_x,\lambda_x)\}_{x\in\Lambda_L}$ for $H_{\Lambda_L}$.

To finish the proof we need to show that $\{(\psi_x,\lambda_x)\}_{x\in\Lambda_L}$ is an $s$-subexponentially localized eigensystem for $\Lambda_L$.
We fix $N=\lfloor Y^s\rfloor$, $x\in\Lambda_L$, take $y\in\Lambda_L$, and consider several cases:

\begin{enumerate}
\item Suppose $\lambda_x\in\sigma_{\mathcal{G}}(\Lambda_L)$. Then $x\in\Lambda_\ell(a_{\lambda_x})$ with $a_{\lambda_x}\in\mathcal{G}$, and
$\lambda_x\in\sigma_{\{a_{\lambda_x}\}}(H_{\Lambda_L})$. In view of \eqref{Lamdecom} we consider two cases:
\begin{enumerate}
\item If $y\in\Lambda_\ell^{\Lambda_L,\frac{\ell}{10}}(a)$ for some $a\in\mathcal{G}$ and $\|y-x\|\geq2\ell$, we must have $\Lambda_\ell(a_{\lambda_x})\cap\Lambda_\ell(a)=\emptyset$,
so it follows from \eqref{sigmaab} that $\lambda_x\not\in  \sigma_{\{a\}}(H_{\Lambda_L})$, and \eqref{psidec1} gives
    \beq\label{psigoodsub}
    |\psi_x|\leq\e^{-c_2\ell^s}|\psi_x(y_1)|\sqtx{for some}y_1\in\partial^{\Theta,2\ell'}\Lambda_\ell(a).
    \eeq
\item If $y\in\Upsilon_r^{\Lambda_L,\frac{\ell}{10}}$ for some $r\in\{1,2,\ldots,R\}$, and $\|y-x\|\geq\ell+\diam\Upsilon_r$, we must have $\Lambda_\ell(a_{\lambda_x})\cap\Upsilon_r=\emptyset$,
so it follows from \eqref{sigmaab} that $\lambda_x\not\in\sigma_{\mathcal{G}_{\Upsilon_r}}(H_{\Lambda_L})$, and clearly $\lambda_x\not\in\sigma_{\Upsilon_r}(H_{\Lambda_L})$
in view of \eqref{sigmasigma}. Thus Lemma~\ref{lembadall}(ii) gives
    \beq\label{psibadsub}
    |\psi_x(y)|\leq\e^{-c_4\ell^s}|\psi_x(v)|\qtx{for some}v\in\partial^{\Lambda_L,2\ell'}\Upsilon_r.
    \eeq
\end{enumerate}
\item Suppose $\lambda_x\not\in\sigma_{\mathcal{G}}(\Lambda_L)$. Then it follows from \eqref{sigmaclaim} that we must have $\lambda_x\in    \sigma_{\Upsilon_{\widetilde{r}}}(H_{\Lambda_L})$
for some $\widetilde{r}\in\{1,2,\ldots,R\}$. In view of \eqref{Lamdecom} we consider two cases:
    \begin{enumerate}
\item If $y\in\Lambda_{\ell}^{\Lambda_L,\frac{\ell}{10}}(a)$ for some $a\in\mathcal{G}$, and $\|y-x\|\geq\ell+\diam\Upsilon_{\widetilde{r}}$, we must have
$\Lambda_\ell(a)\cap\Upsilon_{\widetilde{r}}=\emptyset$, and \eqref{psidec1} gives \eqref{psigoodsub}.
\item If $y\in\Upsilon_r^{\Lambda_L,\frac{\ell}{10}}$ for some $r\in\{1,2,\ldots,R\}$, and $\|y-x\|\geq\diam\Upsilon_{\widetilde{r}}+\diam\Upsilon_r$, we must have $r\neq\widetilde{r}$.
Thus Lemma~\ref{lembadall}(ii) gives \eqref{psibadsub}.
\end{enumerate}
\end{enumerate}

Now we fix $x\in\Lambda_L$, and take $y\in\Lambda_L $ such that $\|y-x\|\geq L'$. Suppose $|\psi_x(y)|>0$ without loss of generality.  We estimate $|\psi_x(y)| $ using either
\eqref{psigoodsub} or \eqref{psibadsub} repeatedly, as appropriate, stopping when we get too close to $x$ so we are not in any case described above. (Note that this must happen since
$|\psi_x(y)|>0$.) We accumulate decay only when we use \eqref{psigoodsub}, and just use $\e^{-c_4\ell^s}<1$ when using \eqref{psibadsub}, recalling $L=Y\ell$, then we get
\beq\label{repeatdecsub}
|\psi_x(y)|\leq\left(\e^{-c_2\ell^s}\right)^{n(Y)},
\eeq
where $n(Y)$ is the number of times we used \eqref{psigoodsub}. We have
\beq
n(Y)(\ell+1)+\sum_{r=1}^R\diam\Upsilon_r+2\ell\geq L'.
\eeq
Thus, using \eqref{sumdiam}, we have
\beq
n(Y)\geq\tfrac{1}{\ell+1}(L'-5\ell\lfloor Y^s\rfloor-2\ell)\geq\tfrac{\ell}{\ell+1}\left(\tfrac{Y}{40}-5Y^s-2\right)\geq2Y^s.
\eeq
for sufficiently large $\ell$ since $Y\geq400^{\frac{1}{1-s}}$. It follows from \eqref{repeatdecsub},
\beq
|\psi_x(y)|\leq\left(\e^{-c_2\ell^s}\right)^{2Y^s}\leq\e^{-L^s},
\eeq
for sufficiently large $\ell$.

We conclude that $\{(\psi_x,\lambda_x)\}_{x\in\Lambda_L}$ is an $s$-subexponentially localized eigensystem for $\Lambda_L$, so the box $\Lambda_L$ is $s$-SEL for $H_{\varepsilon,\omega}$.

\end{proof}

\begin{proof}[Proof of Proposition~\ref{propmsa3}]
We assume \eqref{initcon3} and set $L_{k+1}=YL_k$ for $k=0,1,\ldots$. We set
\beq
\widetilde{P}_k=\sup_{x\in\R^d}\P\{\Lambda_{L_k}(x)\sqtx{is not}s\text{-SEL for}\;H_{\varepsilon,\omega}\}\sqtx{for}k=1,2,\ldots.
\eeq
Then by Lemma~\ref{indumsa3}, we have
\beq\label{calpksub}
\widetilde{P}_{k+1}\leq(2Y)^{(\lfloor Y^s\rfloor+1)d}\widetilde{P}_k^{\lfloor Y^s\rfloor+1}+\tfrac{1}{2}\e^{-L_{k+1}^\zeta}\qtx{for}k=0,1,\ldots
\eeq
If $\widetilde{P}_k\leq\e^{-L_k^\zeta}$ for some $k\geq0$, we have
\begin{align}
\widetilde{P}_{k+1}&\leq(2Y)^{(\lfloor Y^s\rfloor+1)d}\left(\e^{-L_k^\zeta}\right)^{\lfloor Y^s\rfloor+1}+\tfrac{1}{2}\e^{-L_{k+1}^\zeta}
\\\nonumber&\leq(2Y)^{(\lfloor Y^s\rfloor+1)d}\e^{-\frac{\lfloor Y^s\rfloor+1}{Y^\zeta}L_{k+1}^\zeta}+\tfrac{1}{2}\e^{-L_{k+1}^\zeta}\leq\e^{-L_{k+1}^\zeta}
\end{align}
for $L_0$ sufficiently large, since $\zeta<s$. Therefore to finish the proof, we need to show that
\beq
K_0=\inf\{k\in\N;\widetilde{P}_k\leq\e^{-L_k^{\zeta}}\}<\infty.
\eeq

It follows from \eqref{calpksub} that for any $1\leq k<K_0$,
\beq
\widetilde{P}_k\leq(2Y)^{(\lfloor Y^s\rfloor+1)d}\widetilde{P}_{k-1}^{\lfloor Y^s\rfloor+1}+\tfrac{1}{2}\e^{-L_k+^\zeta}
<(2Y)^{(\lfloor Y^s\rfloor+1)d}\widetilde{P}_{k-1}^{\lfloor Y^s\rfloor+1}+\tfrac{1}{2}\widetilde{P}_k,
\eeq
so
\beq
\left(2(2Y)^{(\lfloor Y^s\rfloor+1)d}\right)^{\frac{1}{\lfloor Y^s\rfloor}}\widetilde{P}_k
<\left(\left(2(2Y)^{(N+1)d}\right)^{\frac{1}{\lfloor Y^s\rfloor}}\widetilde{P}_{k-1}\right)^{\lfloor Y^s\rfloor+1}.
\eeq
For $1\leq k<K_0$, since $\left(2(2Y)^{(\lfloor Y^s\rfloor+1)d}\right)^{\frac{1}{\lfloor Y^s\rfloor}}\widetilde{P}_0<1$, we have
\begin{align}\label{ineqpksub}
&\left(2(2Y)^{(\lfloor Y^s\rfloor+1)d}\right)^{\frac{1}{\lfloor Y^s\rfloor}}\e^{-Y^{k\zeta}L_0^\zeta}
=\left(2(2Y)^{(\lfloor Y^s\rfloor+1)d}\right)^{\frac{1}{\lfloor Y^s\rfloor}}\e^{-L_k^\zeta}
\\\nonumber&\quad<\left(2(2Y)^{(\lfloor Y^s\rfloor+1)d}\right)^{\frac{1}{\lfloor Y^s\rfloor}}\widetilde{P}_k
<\left(\left(2(2Y)^{(\lfloor Y^s\rfloor+1)d}\right)^{\frac{1}{\lfloor Y^s\rfloor}}\widetilde{P}_0\right)^{(\lfloor Y^s\rfloor+1)^k}
\\\nonumber&\quad\quad\leq\left(\left(2(2Y)^{(\lfloor Y^s\rfloor+1)d}\right)^{\frac{1}{\lfloor Y^s\rfloor}}\widetilde{P}_0\right)^{Y^{ks}}.
\end{align}
Since $\zeta<s$, $\left(2(2Y)^{(\lfloor Y^s\rfloor+1)d}\right)^{\frac{1}{\lfloor Y^s\rfloor}}\widetilde{P}_0<1$, \eqref{ineqpksub}
cannot be satisfied for large $k$. We conclude that $K_0<\infty$.
\end{proof}

\subsection{The second intermediate step}

\begin{proposition}\label{indumsa3two}
Fix $\varepsilon_0>0$. Suppose for some scale $\ell$ and $0<\varepsilon\leq\varepsilon_0$ we have
\beq\label{hypmsaind3two}
\inf_{x\in\R^d}\P\{\Lambda_\ell(x)\sqtx{is}s\text{-SEL for}\;H_{\varepsilon,\omega}\}\geq1-\e^{-\ell^\zeta}.
\eeq
Then, if $\ell$ is sufficiently large, for $L=\ell^\gamma$ we have
\beq\label{restwo}
\inf_{x\in\R^d}\P\{\Lambda_L(x)\sqtx{is}m_0\text{-localizing for}\;H_{\varepsilon,\omega}\}\geq1-\e^{-L^\zeta},
\eeq
where
\beq\label{m000}
m_0\geq\tfrac{1}{8}L^{-\left(1-\tau+\frac{1-s}{\gamma}\right)}.
\eeq
\end{proposition}

\begin{proof}
We let $\mathcal{B}_N$, $\mathcal{S}_N$ and $\mathcal{E}_N$ as in the proof of Lemma~\ref{indumsa3}. We proceed as in the proof of Lemma~\ref{indumsa3}.
Using \eqref{hypmsaind3two}, \eqref{covernum} and the fact that events on disjoint boxes are independent, we have
\begin{align}\label{probbsubtwo}
\P\{\mathcal{B}^c\}&\leq\left(\tfrac{2L}{\ell}\right)^{(N+1)d}\e^{-(N+1)\ell^\zeta}=2^{(N+1)d}\ell^{(\gamma-1)(N+1)d}\e^{-(N+1)\ell^\zeta}
\\\nonumber&<\tfrac{1}{2}\e^{-\ell^{\gamma\zeta}}=\tfrac{1}{2}\e^{-L^\zeta},
\end{align}
if $N+1>\ell^{(\gamma-1)\zeta}$ and $\ell$ is sufficiently large. For this reason we take
\beq\label{numNsubtwo}
N=N_\ell=\left\lfloor\ell^{(\gamma-1)\widetilde{\zeta}}\right\rfloor\Longrightarrow\P\{\mathcal{B}_{N_\ell}^c\}
\leq\tfrac{1}{2}\e^{-L^\zeta}\qtx{for all}\ell\quad\text{sufficiently large}.
\eeq
Also, using \eqref{problssub} and \eqref{cFN}, we have,
\beq\label{probssubtwo}
\P\{\mathcal{S}_N^c\}\leq Y_{\varepsilon_0}\left(1+(L+1)^dN_\ell!(d4^d)^{N|\ell-1}\right)(L+1)^{2d}\e^{-(2\alpha-1)L^\beta}<\tfrac{1}{2}\e^{-L^\zeta}
\eeq
for sufficiently large $L$, since $(\gamma-1)\widetilde{\zeta}<(\gamma-1)\beta<\gamma\beta$ and $\zeta<\beta$. Combining \eqref{probbsubtwo}
and \eqref{probssubtwo}, we conclude that
\beq
\P\{\mathcal{E}_N\}>1-\e^{-L^\zeta}.
\eeq

To finish the proof we need to show that for all $\omega\in\mathcal{E}_N$ the box $\Lambda_L$ is $m_0$-localizing for $H_{\varepsilon,\omega}$,
where $m_0$ is given in \eqref{m000}. Following the proof of Lemma~\ref{indumsa3}, we get
$\sigma(H_{\Lambda_L})=\sigma_{\mathcal{G}}(H_{\Lambda_L})\cup\sigma_{\mathcal{B}}(H_{\Lambda_L})$ and obtain an eigensystem
$\{(\psi_x,\lambda_x)\}_{x\in\Lambda_L}$ for $H_{\Lambda_L}$. To finish the proof we need to show that $\{(\psi_x,\lambda_x)\}_{x\in\Lambda_L}$
is an $m_0$-localized eigensystem for $\Lambda_L$. We proceed as in the proof of Lemma~\ref{indumsa3}. We fix $x\in\Lambda_L$, and take $y\in\Lambda_L$
such that $\|y-x\|\geq L_\tau$, we have
\beq
n(\ell)(\ell+1)+\sum_{r=1}^R\diam\Upsilon_r+2\ell\geq L_\tau.
\eeq
where $n(\ell)$ is the number of times we used \eqref{psigoodsub}.
Thus, recalling $N=\lfloor\ell^{(\gamma-1)\widetilde{\zeta}}\rfloor$ and using \eqref{sumdiam}, we have
\beq
n(\ell)\geq\tfrac{1}{\ell+1}(L_\tau-5\ell\lfloor\ell^{(\gamma-1)\widetilde{\zeta}}\rfloor-2\ell)
\geq\tfrac{\ell}{\ell+1}\left(\tfrac{1}{2}\ell^{\gamma\tau-1}-5\ell^{(\gamma-1)\widetilde{\zeta}}-2\right)\geq\tfrac{1}{4}\ell^{\gamma\tau-1}.
\eeq
for sufficiently large $\ell$ since $(\gamma-1)\widetilde{\zeta}+1<\gamma\tau$. It follows from \eqref{repeatdecsub},
\begin{align}
|\psi_x(y)|&\leq\left(\e^{-c_2\ell^s}\right)^{\frac{1}{4}\ell^{\gamma\tau-1}}
\\\nonumber&\leq\e^{-\frac{1}{8}L^{-\left(1-\tau+\frac{1-s}{\gamma}\right)}\|y-x\|}
\end{align}
for sufficiently large $\ell$.

We conclude that $\{(\psi_x,\lambda_x)\}_{x\in\Lambda_L}$ is an $m_0$-localized eigensystem for $\Lambda_L$, where $m_0$ is given in \eqref{m000},
so the box $\Lambda_L$ is $m_0$-localizing for $H_{\varepsilon,\omega}$.
\end{proof}

\subsection{The fourth multiscale analysis}

\begin{proposition}\label{propmsa4}
Fix $\varepsilon_0>0$. There exists a finite scale $\mathcal{L}(\varepsilon_0)$ with the following property: Suppose for some scale
$L_0\geq\mathcal{L}(\varepsilon_0)$, $0<\varepsilon\leq\varepsilon_0$, and $m_0\geq L_0^{-\kappa}$,  where $0<\kappa<\tau-\gamma\beta$, we have
\beq\label{initcon4}
\inf_{x\in\R^d}\P\{\Lambda_{L_0}(x)\sqtx{is}m_0\text{-localizing for}\;H_{\varepsilon,\omega}\}\geq1-\e^{-L_0^\zeta}.
\eeq
Then, setting $L_{k+1}=L_k^\gamma$ for $k=0,1,\ldots$, we have
\beq
\inf_{x\in\R^d}\P\{\Lambda_{L_k}(x)\sqtx{is}\tfrac{m_0}{2}\text{-localizing for}\;H_{\varepsilon,\omega}\}\geq1-\e^{-L_k^\zeta}\sqtx{for}k=0,1,\ldots.
\eeq
Moreover, we have
\beq\label{res4}
\inf_{x\in\R^d}\P\{\Lambda_{L_k}(x)\sqtx{is}\tfrac{m_0}{4}\text{-localizing for}\;H_{\varepsilon,\omega}\}\geq1-\e^{-L_k^\xi}\sqtx{for all}L\geq L_0^\gamma.
\eeq
\end{proposition}

\begin{lemma}\label{indumsa4}
Fix $\varepsilon_0>0$. Suppose for some scale $\ell$, $0<\varepsilon\leq\varepsilon_0$, and $m\geq\ell^{-\kappa}$, where $0<\kappa<\tau-\gamma\beta$, we have
\beq\label{hypmsaind4}
\inf_{x\in\R^d}\P\{\Lambda_\ell(x)\sqtx{is}m\text{-localizing for}\;H_{\varepsilon,\omega}\}\geq1-\e^{-\ell^\zeta}.
\eeq
Then, if $\ell$ is sufficiently large, for $L=\ell^\gamma$ we have
\beq
\inf_{x\in\R^d}\P\{\Lambda_L(x)\sqtx{is}M\text{-localizing for}\;H_{\varepsilon,\omega}\}\geq1-\e^{-L^\zeta},
\eeq
where
\beq\label{biqM}
M\geq m\left(1-C_{d,\varepsilon_0}\ell^{-\min\left\{\frac{1-\tau}{2},\gamma\tau-(\gamma-1)\widetilde{\zeta}-1,\tau-\gamma\beta-\kappa\right\}}\right)\geq\tfrac{1}{L^\kappa}.
\eeq
\end{lemma}

Lemma~\eqref{indumsa4} and Proposition~\eqref{propmsa4} follow from \cite[Lemma~4.5]{EK}, \cite[Proposition~4.3]{EK}, and \cite[Section~4.3]{EK}.
(Note that in \cite{EK}, they assume $m\geq m_-$ for a fixed $m_-$. However, all the results still hold when $m\geq\ell^{-\kappa},0<\kappa<\tau-\gamma\beta$.
(See the Lemmas for $\sharp$ being LOC in Sections~\ref{sectlem} and \ref{sectbuff}.))

\subsection{The proof of the bootstrap multiscale analysis}

To prove Theorem~\ref{mainthm}, first we assume \eqref{resmain}, which is the same as \eqref{initcon} with letting $Y=400$, for some length scales. We apply Proposition~\ref{propmsa1}, obtaining a sequence of length scales satisfying \eqref{res1}.
Therefore \eqref{hypmsaindone} is satisfied for some length scales. Applying Proposition~\ref{indumsa1one}, we get a length scale satisfying \eqref{resone}. It follows that \eqref{initcon2} is satisfied since $0<1-\tau+\tfrac{1}{\gamma_1}<\tau$.
We apply Proposition~\ref{propmsa2}, obtaining a sequence of length scales satisfying \eqref{res2}. Therefore, In view of Remark~\ref{rmk1}, \eqref{initcon3} is satisfied with letting $Y=400^{\frac{1}{1-s}}$. We apply Proposition~\ref{propmsa3}, obtaining a sequence of length scales satisfying \eqref{res3}.
Therefore \eqref{hypmsaind3two} is satisfied for some length scales. Applying Proposition~\ref{indumsa3two}, we get a length scale satisfying \eqref{restwo}. It follows that \eqref{initcon4} is satisfied since $0<1-\tau+\tfrac{1-s}{\gamma}<\tau-\gamma\beta$.  We apply Proposition~\ref{propmsa4}, getting \eqref{res4}, so \eqref{resmain} holds.

\section{The initial step for the bootstrap multiscale analysis}\label{sectinit}

Theorem~\ref{maincor} is an immediate consequence of Theorem~\ref{mainthm} and Proposition~\ref{propinit}.

\begin{proposition}\label{propinit}
Given $q>\tfrac{2d}{\alpha}$ and  $\varepsilon>0$, set
\beq\label{thetaprop}
\theta_{\varepsilon,L}=\frac{\left\lfloor\frac{L}{20}\right\rfloor}{\log L}\log\left(1+\frac{L^{-q}}{2d\varepsilon}\right).
\eeq
Then
\begin{align}\label{probploc1}
&\inf_{x\in\R^d}\P\{\Lambda_L(x)\sqtx{is}\theta_{\varepsilon,L}\text{-polynomially localizing for}\;H_{\varepsilon,\omega}\}
\\\nonumber&\hspace{120pt}\geq1-\tfrac{1}{2}K(L+1)^{2d}\left(8d\varepsilon+2L^{-q}\right)^\alpha.
\end{align}
In particular, given $\theta>0$ and $P_0>0$, there exists a finite scale $\mathcal{L}(q,\theta,P_0)$ such that for all $L\geq\mathcal{L}(q,\theta,P_0)$ and $0<\varepsilon\leq \tfrac{1}{4d}L^{-q}$ we have
\begin{align}\label{probploc2}
 \inf_{x\in\R^d}\P\{\Lambda_L(x)\sqtx{is}\theta\text{-polynomially localizing for}\;H_{\varepsilon,\omega}\}
\geq   1-P_0.
\end{align}
\end{proposition}

Proposition~\ref{propinit} shows that the starting hypothesis for the bootstrap multiscale analysis of Theorem~\ref{mainthm}  can be fulfilled .

To prove Proposition~\ref{propinit}, we will use the following lemma given in \cite[Lemma~4.4]{EK}.
\begin{lemma}[{\cite[Lemma~4.4]{EK}}]\label{leminit}
Let $H_\varepsilon=-\varepsilon\Delta+V$ on $\ell^2(\Z^d)$, where $V$ is a bounded potential and $\varepsilon>0$. Let $\Theta\subset\Z^d$, and suppose there is $\eta>0$ such that
\beq\label{Vdif}
|V(x)-V(y)|\geq\eta\qtx{for all}x,y\in\Theta, x\neq y.
\eeq
Then for $\varepsilon<\frac{\eta}{4d}$ the operator $H_{\varepsilon,\Theta}$ has an eigensystem $\{(\psi_x,\lambda_x)\}_{x\in\Theta}$ such that
\beq\label{lambdadif}
|\lambda_x-\lambda_y|\geq\eta-4d\varepsilon>0\qtx{for all}x,y\in\Theta,x\neq y,
\eeq
and for all $y\in\Theta$ we have
\beq\label{initdec}
|\psi_y(x)|\leq\left(\tfrac{2d\varepsilon}{\eta-2d\varepsilon}\right)^{|x-y|_1}\qtx{for all}x\in\Theta.
\eeq
\end{lemma}

\begin{proof}[Proof of Proposition~\ref{propinit}]
Let $\varepsilon>0$ and $\Lambda_L=\Lambda_L(x_0)$ for some $x_0\in\R^d$. Let $\eta=4d\varepsilon+L^{-q}$ and suppose
\beq\label{VdiffL}
|V(x)-V(y)|\geq\eta\qtx{for all}x,y\in\Theta, x\neq y.
\eeq
It follows from Lemma~\ref{leminit} that $H_{\varepsilon,\Lambda_L}$ has an eigensystem $\{(\psi_x,\lambda_x)\}_{x\in\Lambda_L}$ satisfying
\eqref{lambdadif} and \eqref{initdec}. We conclude from \eqref{lambdadif} that $\Lambda_L$ is polynomially level spacing for $H_\varepsilon$.
Moreover, using \eqref{initdec} and $\|x\|\leq|x|_1$, for all $y,x\in\Lambda_L$ with $\|x-y\|\geq L'$ we have
\begin{align}
|\psi_y(x)|&\leq\left(\tfrac{2d\varepsilon}{\eta-2d\varepsilon}\right)^{\|x-y\|}=L^{-\frac{\|x-y\|}{\log L}\log\left(\frac{\eta-2d\varepsilon}{2d\varepsilon}\right)}
\\\nonumber&=L^{-\frac{\|x-y\|}{\log L}\log\left(1+\frac{L^{-q}}{2d\varepsilon}\right)}\leq L^{-\theta_{\varepsilon,L}}
\end{align}
with $\theta_{\varepsilon,L}$ as in \eqref{thetaprop}. Therefore $\Lambda_L(x)$ is $\theta$-polynomially localizing.

We have
\begin{align}
&\P\{\Lambda_L\sqtx{is not}\theta_{\varepsilon,L}\text{-polynomially localizing}\}\leq\P\{\text{\eqref{VdiffL} does not hold}\}
\\\nonumber&\qquad\qquad\leq\tfrac{(L+1)^{2d}}{2}S_\mu\left(2\left(4d\varepsilon+L^{-q}\right)\right)\leq\tfrac{1}{2}K{(L+1)^{2d}}\left(8d\varepsilon+2L^{-q}\right)^\alpha,
\end{align}
which yields \eqref{probploc1}. (We assumed $8d\varepsilon+2\L^{-q}\leq1$; if not \eqref{probploc1} holds trivially.)

If $0<\varepsilon\leq\tfrac{1}{4d}L^{-q}$, for sufficiently large $L$ we have $\theta_{\varepsilon,L}\geq\theta$, and
\beq
\inf_{x\in\R^d}\P\{\Lambda_L(x)\sqtx{is}\theta\text{-polynomially localizing for}\;H_{\varepsilon,\omega}\}
\geq1-P_0,
\eeq
since $\alpha q-2d>0$.
\end{proof}

\smallskip
\begin{acknowledgement}  A.K. wants to thank Alexander Elgart for many discussions.
\end{acknowledgement}

\end{document}